\documentclass[envcountsame,runningheads]{llncs}
\usepackage[T1]{fontenc}
\usepackage{amsmath}
\usepackage{microtype}

\title{Three-element \minsol\ and Conservative \minhom}
\author{Hannes Uppman%
\thanks{Partially supported by the National Graduate School in Computer Science (CUGS), Sweden.}}
\institute{%
Department of Computer and Information Science, \\
Link{\"{o}}ping University, SE-581 83 Link{\"{o}}ping, Sweden \\
\email{hannes.uppman@liu.se}}

\def\clap#1{\hbox to 0pt{\hss#1\hss}}
\def\mathllap{\mathpalette\mathllapinternal}

\def\mathclap{\mathpalette\mathclapinternal}
\def\mathllapinternal#1#2{\llap{$\mathsurround=0pt#1{#2}$}}

\def\mathclapinternal#1#2{\clap{$\mathsurround=0pt#1{#2}$}}

\newcommand\myvcenter[1]{\ensuremath{\vcenter{\hbox{#1}}}}

\makeatletter
\newcommand\st {\@ifnextchar{~}{s.t.}{s.t.\ }}
\newcommand\ie {\@ifnextchar{~}{i.e.}{i.e.\ }}
\newcommand\eg {\@ifnextchar{~}{e.g.}{e.g.\ }}
\newcommand\wlg{\@ifnextchar{~}{wlog}{wlog\ }}
\makeatother

\renewcommand\phi\varphi
\renewcommand\rho\varrho
\renewcommand\epsilon\varepsilon

\newcommand\propref[1]{Proposition~\ref{#1}}

\newcommand\lemref[1]{Lemma~\ref{#1}}
\newcommand\corref[1]{Corollary~\ref{#1}}
\newcommand\thmref[1]{Theorem~\ref{#1}}
\newcommand\thmrefs[1]{Theorems~\ref{#1}}
\newcommand\sectref[1]{Sect.~\ref{#1}}
\newcommand\sectsref[1]{Sects.~\ref{#1}}
\newcommand\startsectref[1]{Section~\ref{#1}}

\newcommand\domp[2]{$(#1,#2)$-dominating}

\newcommand\typegwtp{{\bf GWTP}}
\newcommand\typebsm{{\bf BSM}}
\newcommand\typegmc{{\bf GMC}}

\newcommand{\csp}{CSP}
\newcommand{\vcsp}{VCSP}
\newcommand\minhom{Min-Cost-Hom}
\newcommand\minsol{Min-Sol}
\newcommand\minones{Min-Ones}

\newcommand\class[1]{#1}

\newcommand\close[1] {\langle #1 \rangle}
\newcommand\wclose[1]{\langle #1 \rangle_w}
\newcommand\eclose[1]{\langle #1 \rangle_e}

\newcommand\binopers{\smash{\mathcal O_D^{(2)}}}
\newcommand\opers[1]{\mathcal O_D^{(#1)}}

\newcommand\qplus{\bbbq_{\smash{\ge 0}}}

\DeclareMathOperator*\argmin{\arg\,\min}

\DeclareMathOperator\ar{ar}
\DeclareMathOperator\pr{pr}

\DeclareMathOperator\supp{supp}
\DeclareMathOperator\sol{Sol}
\DeclareMathOperator\opt{Opt}
\DeclareMathOperator\optsol{Optsol}
\DeclareMathOperator\pol{Pol}
\DeclareMathOperator\fpol{f{}Pol}

\def\ba#1\ea{\begin{align*}#1\end{align*}}
\def\bna#1\ena{\begin{align}#1\end{align}}
\newcommand\bi{\begin{itemize}}
\newcommand\ei{\end{itemize}}
\newcommand\be{\begin{enumerate}}
\newcommand\ee{\end{enumerate}}

\newcommand\ssum[1]{\sum_{\mathclap{#1}}}

\newcommand\labels{%
\put(-5,-0.5){\scriptsize{a}}
\put(-5,8){\scriptsize{b}}
\put(-5,17.5){\scriptsize{c}}
\put(11,-0.5){\scriptsize{a}}
\put(11,8){\scriptsize{b}}
\put(11,17.5){\scriptsize{c}}
}
\newcommand\mypic[1]{\myvcenter{%
\begin{picture}(22,20)(-6,0)#1\labels\end{picture}}}

\newcounter{tmpb}
\newcounter{tmpa}
\newcommand\makeline[2]{
\setcounter{tmpb}{#2}  \addtocounter{tmpb}{-#1} 
\setcounter{tmpa}{#1}  \addtocounter{tmpa}{#1}
\addtocounter{tmpa}{#1}\addtocounter{tmpa}{#1}
\addtocounter{tmpa}{#1}\addtocounter{tmpa}{#1}
\addtocounter{tmpa}{#1}\addtocounter{tmpa}{#1}
\addtocounter{tmpa}{#1}\addtocounter{tmpa}{#1}  
\put(0,\value{tmpa}){\line(1,\value{tmpb}){10}}
}

\newcommand\mysmallpic[5]{%
\myvcenter{\begin{picture}(0 ,10)(0,0)  \end{picture}}^{#1}_{#2}
\myvcenter{\begin{picture}(10,10)(0,0)#5\end{picture}}^{\;\!#3}_{\;\!#4}
}
\newcommand\smallcrosspic[4]{\mysmallpic{#1}{#2}{#3}{#4}{
\put(0,0){\line(1,1){10}}
\put(0,10){\line(1,-1){10}}
}}
\newcommand\smallmispic[4]{\mysmallpic{#1}{#2}{#3}{#4}{
\put(0,0){\line(1,1){10}}
\put(0,10){\line(1,-1){10}}
\put(0,10){\line(1,0){10}}
}}
\newcommand\smallflipmispic[4]{\mysmallpic{#1}{#2}{#3}{#4}{
\put(0,0){\line(1,1){10}}
\put(0,10){\line(1,-1){10}}
\put(0,0){\line(1,0){10}}
}}

\newcommand\smallallpic[4]{\mysmallpic{#1}{#2}{#3}{#4}{
\put(0,0){\line(1,0){10}}
\put(0,0){\line(1,1){10}}
\put(0,10){\line(1,-1){10}}
\put(0,10){\line(1,0){10}}
}}

\newcommand\upasymbol{{\scriptstyle\uparrow}}
\newcommand\doasymbol{{\scriptstyle\downarrow}}
\newcommand\ddsymbol{{\scriptstyle\diamondsuit}}

\newcommand\upa[2]{\mathop\upasymbol\limits_{\mathclap{\vphantom{c}\smash{#2}}}^{\mathclap{\vphantom{c}\smash{#1}}}}
\newcommand\doa[2]{\mathop\doasymbol\limits_{\mathclap{\vphantom{c}\smash{#2}}}^{\mathclap{\vphantom{c}\smash{#1}}}}
\newcommand\dd[3]{\mathop{\ddsymbol_{#3}}\limits_{\mathclap{{#2}_{#3}}}^{\mathclap{{#1}_{#3}}}}

\newcommand\domps[1]{\triangleright_{#1}}

\begin{document}

\maketitle

\begin{abstract}
Thapper and \v{Z}ivn\'{y} [STOC'13] recently classified the complexity of \vcsp\ for all finite-valued constraint languages. 
However, the complexity of {\vcsp}s for constraint languages that are not finite-valued remains poorly understood.
In this paper we study the complexity of two such {\vcsp}s, namely \minhom\ and \minsol.
We obtain a full classification for the complexity of \minsol\ on domains that contain at most three elements and for the complexity of conservative \minhom\ on arbitrary finite domains.
Our results answer a question raised by Takhanov [STACS'10, COCOON'10].
\end{abstract}

\section {Introduction}
The \emph{valued constraint satisfaction problem} (\vcsp)
is a very broad framework in which many combinatorial optimisation problems can be expressed.
A \emph{valued constraint language} is a fixed set of cost functions from powers of a \emph{finite domain}.
An instance of \vcsp\ for some give constraint language is then a weighted sum of cost functions from the language.
The goal is to minimise this sum.
On the two-element domain the complexity of the problem is known for every constraint language~\cite{softcsp}.
Also for every language containing all $\{0,1\}$-valued unary cost functions the complexity is known~\cite{kolmogorov:zivny:conservativevcsp}.
In a recent paper Thapper and \v{Z}ivn\'{y}~\cite{thapper:zivny:fvcsp} managed to classify the complexity of \vcsp\ for all finite-valued constraint languages.
However, {\vcsp}s with other types of languages remains poorly understood.

In this paper we study the complexity of the \emph{(extended) minimum cost homomorphism problem} (\minhom) and the \emph{minimum solution problem} (\minsol).
These problems are both {\vcsp}s with special types of languages in which all non-unary cost-functions are crisp ($\{0,\infty\}$-valued).
Despite this rather severe restriction the frameworks allow many natural combinatorial optimisation problems to be expressed.
\minsol\ does \eg generalise a large class of bounded integer linear programs.
It may also be viewed as a generalisation of the problem \minones~\cite{minones} to larger domains.
The problem \minhom\ is even more general and contains \minsol\ as a special case.

The problem \minsol\ has received a fair bit of attention in the literature and has \eg had its complexity fully classified for all graphs of size three~\cite{maxsolgraph} and for all so-called homogeneous languages~\cite{maxonesgen}.
For more information about \minsol\ see~\cite{intromaxsol} and the references therein.
The ``unextended version'' of \minhom\ was introduced in~\cite{minhom:1} motivated by a problem in defence logistics.
It was studied in a series of papers before it was completely solved in~\cite{rustem:minhom}.
The more general version of the problem which we are interested in was introduced in~\cite{rustem:minsol}.%
\footnote{The definition in~\cite{rustem:minsol} is slightly more restrictive than the one we use.
Also the notation differs; what we denote \minhom$(\Gamma,\Delta)$ is in~\cite{rustem:minsol} referred to as $MinHom_\Delta(\Gamma)$.}

\bigskip
\noindent
{\bf Methods and Results.}
We obtain a full classification of the complexity of \minsol\ on domains that contain at most three elements.
The tractable cases are given by languages that can be solved by a certain linear programming formulation~\cite{thapper:zivny:lp} and a new class that is inspired by, and generalises, languages described in~\cite{rustem:minhom,rustem:minsol}.
A precise classification is given by \thmref{thm:minsol:2}.
For conservative \minhom\ (\ie \minhom\ with languages containing all unary crisp cost functions)
an almost complete classification (for arbitrary finite domains) was obtained by Takhanov~\cite{rustem:minsol}.
We are able to remove the extra conditions needed in~\cite{rustem:minsol} and provide a full classification for this problem.
This answers a question raised in~\cite{rustem:minhom,rustem:minsol}.
The main mathematical tools used througout the paper are from the so-called algebraic approach, see \eg~\cite{alg:3,alg:1},
and its extensions to optimisation problems~\cite{alg:opt:1,alg:opt:2}.
Following~\cite{thapper:zivny:fvcsp} we also make use of Motzkin's Transposition Theorem from the theory of linear equations.

The rest of the paper is organised as follows.
\startsectref{s:prelim} contains needed concepts and results from the literature and
\sectref{s:contrib} holds the description of our results.
Proofs of theorems in \sectref{s:contrib} are given in
\sectsref{s:proof:1},~\ref{s:proof:2} and~\ref{s:proof:3}.
One of our theorems is proved with the help of a fairly lengthy case analysis.
The proofs of this result and two supporting lemmas are collected in three appendices.

\section {Preliminaries}
\label{s:prelim}
For a set $\Gamma$ of finitary relations on a finite set $D$ (the domain),
and a finite set $\Delta$ (referred to as the domain valuations) of functions $D \to \qplus \cup \{ \infty\}$,
we define \minhom$(\Gamma,\Delta)$ as the following optimisation problem.
\begin{description}
\item[Instance:] A triple $(V,C,w)$ where
\bi
\item
$V$ is a set of variables,
\item
$C$ is a set of $\Gamma$-allowed constraints, \ie a set of pairs $(s,R)$ where the constraint-scope $s$ is a tuple of variables, and the constraint-relation $R$ is a member of $\Gamma$ of the same arity as $s$,
and
\item
$w$ is a weight function $V \times \Delta \to \qplus$.
\ei
\item[Solution:] A function $\phi : V \to D$ \st for every $(s,R) \in C$ it holds that $\phi(s) \in R$, where $\phi$ is applied component-wise.
\item[Measure:] The measure of a solution $\phi$ is $m(\phi) = \sum_{v \in V} \sum_{\nu \in \Delta} w(v,\nu) \nu( \phi(v) )$.
\end{description}
The objective is to find a solution $\phi$ that minimises $m(\phi)$.

The problem \minsol$(\Gamma,\nu)$, which we define only for injective functions $\nu : D \to \qplus$, is the problem \minhom$(\Gamma,\{\nu\})$.
The ``regular'' \emph{constraint satisfaction problem} (\csp) can also be defined through \minhom;
an instance of \csp$(\Gamma)$ is an instance of \minhom$(\Gamma,\emptyset)$, and the objective is to determine if any solution exists.

We will call the pair $(\Gamma,\Delta)$ a \emph{language} (or a \minhom-language).
The language \smash{$(\Gamma,\{\nu\})$ is written $(\Gamma,\nu)$}.
For an instance $I$ we use $\opt(I)$ for the measure of an optimal solution (defined only if a solution exists),
$\sol(I)$ denotes the set of all solutions and $\optsol(I)$ the set of all optimal solutions.
We define $0\,\infty = \infty\,0 = 0$ and for all $x \in \qplus \cup \{ \infty\}$, $x \le \infty$ and $x + \infty = \infty + x = \infty$.
The $i$:th projection operation will be denoted $\pr_i$.
We define $\smash{\binom A 2} = \{ \{x,y\} \subseteq A : x \ne y \}$.
$\opers m$ is used for the set of all $m$-ary operations on $D$.
For binary operations $f$, $g$ and $h$ we define $\overline f$ through $\overline f(x,y) = f(y,x)$ and $f[g,h]$ through $f[g,h](x,y) = f(g(x,y),h(x,y))$.
A $k$-ary operation $f$ on $D$ is called \emph{conservative} if $f(x_1,\dots,x_k) \in \{x_1,\dots,x_k\}$ for every $x_1,\dots,x_k \in D$.
A ternary operation $m$ on $D$ is called \emph{arithmetical} on $B \subseteq \binom D 2$ if for every $\{a,b\} \in B$ the function $m$ satisfies $m(a,b,b)=m(a,b,a)=m(b,b,a)=a$.

\bigskip
\noindent
{\bf Polymorphisms.}
Let $(\Gamma,\Delta)$ be a language on the domain $D$.
By $\Gamma^c$ we denote $\Gamma$ enriched with all constants, \ie $\Gamma \cup \{ \{c\} : c \in D \}$.
An operation $f : D^m \to D$ is called a \emph{polymorphism} of $\Gamma$ if for every $R \in \Gamma$ and every sequence $t^1,\dots,t^m \in R$ it holds that $f(t^1,\dots,t^m)\in R$  where $f$ is applied component-wise.
The set of all polymorphisms of $\Gamma$ is denoted $\pol(\Gamma)$.
A function $\omega : \smash{\opers k} \to \qplus$ is a $k$-ary \emph{fractional polymorphism}~\cite{alg:opt:1} of $(\Gamma,\Delta)$ if 
\ba
&\vphantom{\sum_{.}}\smash{
 \ssum{g\in \opers k} \omega(g) = 1 \quad \text{and} \quad \ssum{g\in \opers k} \omega(g) \nu(g(x_1,\dots, x_k)) \le \smash{ \frac{1}{k} \sum_{i=1}^k \nu(x_i) }
}
\ea
holds for every $\nu \in \Delta$ and every $x_1,\dots,x_k \in D$, and $\omega(g) = 0$ if $g \not\in \pol(\Gamma)$.
For a $k$-ary fractional polymorphism $\omega$ we let $\supp(\omega) = \{ g \in \smash{\opers k} : \omega(g) > 0 \}$.
The set of all fractional polymorphisms of $(\Gamma,\Delta)$ is denoted $\fpol(\Gamma,\Delta)$.

\bigskip
\noindent
{\bf Min-cores.}
The language $(\Gamma,\Delta)$ is called a \emph{min-core}~\cite{mincore:2} if there is no non-surjective unary $f \in \pol(\Gamma)$ for which $\nu(f(x)) \le \nu(x)$ holds for every $x \in D$ and $\nu \in \Delta$.
The language $(\Gamma',\Delta')$ is a min-core of $(\Gamma,\Delta)$ if $(\Gamma',\Delta')$ is a min-core and
$(\Gamma,\Delta)|_{f(D)} = (\Gamma',\Delta')$ for some unary $f \in \pol(\Gamma)$ satisfying $\nu(f(x)) \le \nu(x)$ for every $x \in D$ and $\nu \in \Delta$.
The reason why we care about min-cores is the following result~\cite{mincore:2}.%
\footnote{
The results in~\cite{mincore:2} are stated for a slightly more restricted problem than ours.
It is however not hard to see that the results transfer to our setting.
}
\begin{theorem}
Let $(\Gamma',\Delta')$ be a min-core of $(\Gamma,\Delta)$.
If \minhom$(\Gamma',\Delta')$ is \class{NP}-hard (in \class{PO}), then \minhom$(\Gamma,\Delta)$ is \class{NP}-hard (in \class{PO}).
\end{theorem}

\bigskip
\noindent
{\bf Expressive Power and Polynomial-time Reductions.}
A relation $R$ is said to be \emph{weighted pp-definable} in $(\Gamma,\Delta)$ if there is an instance $I=(V,C,w)$ of \minhom$(\Gamma,\Delta)$ \st
$R = \{ (\phi(v_1),\dots,\phi(v_n)) : \phi \in \optsol(I) \}$ for some $v_1,\dots,v_n \in V$.
We use $\wclose{\Gamma,\Delta}$ to denote the set of all relations that is weighted pp-definable in $(\Gamma,\Delta)$.
Similarly $R$ is said to be \emph{pp-definable} in $\Gamma$ if there is an instance $I=(V,C)$ of \csp$(\Gamma)$ \st
$R = \{ (\phi(v_1),\dots,\phi(v_n)) : \phi \in \sol(I) \}$ for some $v_1,\dots,v_n \in V$.
$\close{\Gamma}$ is used to denote the set of all relations that are pp-definable in $\Gamma$.
A cost function $\nu : D \to \qplus \cup \{\infty\}$ is called \emph{expressible} in $(\Gamma,\Delta)$ if there is an instance $I=(V,C,w)$ of \minhom$(\Gamma,\Delta)$ and $v \in V$ \st $\nu(x) = \min \{ m(\phi) : \phi \in \sol(I), \phi(v)=x \}$ if $\nu(x)<\infty$ and 
$\min \{ m(\phi) : \phi \in \sol(I), \phi(v)=x \} = \infty$ or $\{ \phi \in \sol(I) : \phi(v)=x \} = \emptyset$ if $\nu(x)=\infty$.
The set of all cost functions expressible in $(\Gamma,\Delta)$ is denoted $\eclose{\Gamma,\Delta}$.
What makes all these closure operators interesting is the following result, see \eg~\cite{alg:opt:1,softcsp,fourel}.
\begin{theorem}
Let $\Gamma' \subseteq \wclose{\Gamma,\Delta}$ and $\Delta' \subseteq \eclose{\Gamma,\Delta}$ be finite sets.
Then, \minhom$(\Gamma',\Delta')$ is polynomial-time reducible to \minhom$(\Gamma,\Delta)$.
\end{theorem}
This of course also means that if $\Gamma' \subseteq \wclose{\Gamma,\Delta}$ is finite,
then \minhom$(\Gamma' \cup \Gamma,\Delta)$ is polynomial-time reducible to \minhom$(\Gamma,\Delta)$.

We will often use bipartite-graph-representations for relations, \eg $\smallmispic a b b c = \{(a,b), (a, c), (b,b)\}$.
Finally we recall a classic result, see \eg~\cite[p.~94]{schrijver}, about systems of linear equations that will be of great assistance.
\begin{theorem}[Motzkin's Transposition Theorem]
\label{thm:motzkin}
For any
$A \in \bbbq^{m \times n}$,
$B \in \bbbq^{p \times n}$,
$b \in \bbbq^{m}$ and
$c \in \bbbq^{p}$, exactly one of the following holds:
\bi
\item $Ax \le b$, $Bx < c$ for some $x \in \bbbq^n$
\item $A^Ty + B^Tz = 0$ and ($b^T y + c^T z < 0$ or $b^T y + c^T z = 0$ and $z \ne 0$) for some $y \in \qplus^m$ and $z \in \qplus^p$
\ei
\end{theorem}

\section {Contributions}
\label{s:contrib}
We let $D$ denote the finite domain over which the language $(\Gamma,\Delta)$ is defined.
To describe our results we need to introduce some definitions.
\begin{definition}[\domp a b]
\label{def:domp:fpol}
Let $a,b \in D$.
A binary fractional polymorphism $\omega$ of $(\Gamma,\Delta)$ is called \emph{\domp a b} if
\ba
\smash[b]{\ssum{g \in \opers 2} \omega(g) \delta_{a,g(a,b)} \ge \frac{1}{2} > \ssum{g \in \opers 2} \omega(g) \delta_{b,g(a,b)}.}\footnotemark
\ea%
\footnotetext{Here $\delta$ denotes the Kronecker delta function, \ie $\delta_{i,j}=1$ if $i=j$, otherwise $\delta_{i,j}=0$.}
\end{definition}

The following is a generalisation of the concept of weak tournament pairs that was introduced in~\cite{rustem:minsol}.
\begin{definition}[generalised weak tournament pair]
\label{def:wt}
Let $A \subseteq B \subseteq \binom D 2$.
A language $(\Gamma,\Delta)$ is said to admit a \emph{generalised weak tournament pair} on $(A,B)$ if there is a pair of binary functions $f_1,f_2 \in \pol(\Gamma)$ \st the following holds.
\bi
\item
For every $\{a,b\} \in \binom D 2$;
\be
\item \label{wt:a}
if $\{a,b\} \not\in B$ then $f_1|_{\{a,b\}}$ and $f_2|_{\{a,b\}}$ are projections, and
\item \label{wt:b}
if $\{a,b\} \in B \setminus A$ then $f_1|_{\{a,b\}}$ and $f_2|_{\{a,b\}}$ are different idempotent, conservative and commutative operations.
\ee
\item
For any $U \subseteq D$ \st $U \in \close{\Gamma}$ either no $\{x,y\} \in A$ satisfies $\{x,y\} \subseteq U$,
or there is $\{a,b\} \in A$ \st $U \setminus \{b\} \in \close{\Gamma}$ and $(\Gamma,\Delta)$ admits
an \domp a b binary fractional polymorphism.
\ei
\end{definition}
The following definition is inspired by notation used in~\cite{rustem:minhom}.
\begin{definition}
\label{rustemarrow}
For $a,b \in D$ we define $\upa a b = \{ f \in \binopers : f(a,b)=f(b,a)=a \}$ and $\smash{\doa a b} = \{ f \in \binopers : f(a,b)=f(b,a)=b \}$.
For $x_1,\dots,x_m,y_1,\dots,y_m \in D$ and $\ddsymbol_1,\dots,\ddsymbol_m \in \{\upasymbol,\doasymbol\}$ we define
$\dd{x}{y}{1} \dd{x}{y}{2} \cdots \dd{x}{y}{m} = \dd{x}{y}{1} \cap \dd{x}{y}{2} \cap \dots \cap \dd{x}{y}{m}$, \eg $\upa a b \doa c d = \upa a b \cap \doa c d$.
\end{definition}
We can now give names to some classes of languages that will be important.
\begin{definition}
\label{def:names}
We say that a language $(\Gamma,\Delta)$ over $D$ is of type
\bi
\item \typegwtp\ (generalised weak tournament pair) if there is $A,B \subseteq \binom D 2$ \st $(\Gamma,\Delta)$ admits a generalised weak tournament pair on $(A,B)$ and, $\pol(\Gamma)$ contains an idempotent ternary function $m$ that is arithmetical on $\binom D 2 \setminus B$
and satisfies $m(x,y,z) \in \{x,y,z\}$ for every $x,y,z \in D$ \st $|\{x,y,z\}|=3$,
\item \typebsm\ (bisubmodular, see \eg~\cite{softcsp}) if $D = \{a,b,c\}$, $2 \nu(b) \le \nu(a)+\nu(c)$ for every $\nu \in \Delta$, and there are binary idempotent commutative operations $\sqcap,\sqcup \in \pol(\Gamma)$ \st $\sqcap \in \doa a b \doa c b$, $\sqcup \in \upa a b \upa c b$ and $a \sqcup c =  a \sqcap c = b$,
\item \typegmc\ (generalised min-closed, see~\cite{maxonesgen}) if there is $f \in \pol(\Gamma)$ \st for every $\nu \in \Delta$ the following is true.
For all $a,b \in D$ \st $a \ne b$ it holds that if $\nu( f(a,b) ) \ge \max( \nu(a), \nu(b) )$, then $\nu( f(b,a) ) < \min( \nu(a), \nu(b) )$, and for all $a \in D$ it holds that $\nu( f(a,a) ) \le \nu(a)$.
\ei
\end{definition}
Solving instances of \minhom\ expressed in languages of type \typegwtp, \typebsm\ and \typegmc\ can be done in polynomial time.
This is demonstrated by the following results. 
We note that the first result describes a new tractable class while the following two are known cases.%
\footnote{
\cite[Theorem~5.10]{maxonesgen} is stated for a slightly more restricted problem than ours.
It is however not hard to see that the results transfer to our setting.
}
A proof of \thmref{thm:wt:p} is given in \sectref{s:proof:1}.
\begin{theorem}
\label{thm:wt:p}
If there is $S \subseteq 2^D$ \st \csp$(\Gamma^c \cup S)$ is in \class P and $(\Gamma \cup S,\Delta)$ is of type \typegwtp, then \minhom$(\Gamma \cup S, \Delta)$ (and therefore also \minhom$(\Gamma, \Delta)$) is in \class{PO}.
\end{theorem}
\begin{theorem}[{\cite[Corollary~6.1]{thapper:zivny:lp}}]
\label{thm:bsm:sl:p}
If $(\Gamma,\Delta)$ is of type \typebsm, then \minhom$(\Gamma, \Delta)$ is in \class{PO}.
\end{theorem}
\begin{theorem}[{\cite[Theorem~5.10]{maxonesgen}}]
\label{thm:gmc:p}
If $(\Gamma,\Delta)$ is of type \typegmc, then \minhom$(\Gamma, \Delta)$ is in \class{PO}.
\end{theorem}

Instances expressed using languages of type \typebsm\ can, as proved in~\cite{thapper:zivny:lp}, be solved through a certain linear programming formulation.
We note that this also holds for languages of type \typegmc\@.
It is known that any language of type \typegmc\ must admit a min-set-function~\cite[Theorem~5.18]{kuivinen}.
From this it follows that also a symmetric fractional polymorphism of every arity must be admitted, and the claim follows from~\cite{thapper:zivny:lp}.

The tractability of languages of type \typegwtp\ on the other hand can not directly be explained by the results in~\cite{thapper:zivny:lp}.
It can \eg be checked that the language $(\{ \smallcrosspic b a b a \}, \{ a \mapsto 0, b \mapsto 1 \})$ is of type \typegwtp.
This language does not admit any symmetric fractional polymorphism and is therefore not covered by the results in~\cite{thapper:zivny:lp}.

Often (as \eg demonstrated by \thmref{thm:wt:p}) the fact that a language admits an \domp a b binary fractional polymorphism can be useful for tractability arguments.
Also the converse fact, that a language does not admit such a fractional polymorphism, can have useful consequences.
An example of this is the following proposition, which will be used in the proofs of our main results.
\begin{proposition}
\label{prop:use:dom}
Let $a,b \in D$, $a \ne b$.
If $(\Gamma,\Delta)$ does not admit a binary fractional polymorphism that is \domp a b, then $\eclose{\Gamma,\Delta}$ contains a unary function $\nu$ that satisfies $\infty > \nu(a)>\nu(b)$.
\end{proposition}
The proof is given in \sectref{s:proof:2}.

\subsection {Conservative Languages}
We call $(\Gamma,\Delta)$ \emph{conservative} if $2^D \subseteq \Gamma$, \ie if the crisp language contains all unary relations.
The complexity of \minhom$(\Gamma,\Delta)$ for conservative languages $(\Gamma,\Delta)$ was classified in~\cite{rustem:minsol} under the restriction that $\Delta$ contains only finite-valued functions,
and that for each pair $a,b \in D$ there exists some $\nu\in\Delta$ \st either $\nu(a) < \nu(b)$ or $\nu(a) > \nu(b)$.
It was posted in~\cite{rustem:minhom,rustem:minsol} as an open problem to classify the complexity of the problem also without restrictions on $\Delta$.
The following theorem does just that.

\begin{theorem}
\label{thm:conservative}
Let $(\Gamma,\Delta)$ be a conservative language on a finite domain.
If \csp$(\Gamma)$ is in \class{P} and $(\Gamma,\Delta)$ is of type \typegwtp, then \minhom$(\Gamma,\Delta)$ is in \class{PO}, otherwise \minhom$(\Gamma,\Delta)$ is \class{NP}-hard.
\end{theorem}
We prove the theorem in \sectref{s:proof:3}.

Kolmogorov and \v{Z}ivn\'{y}~\cite{kolmogorov:zivny:conservativevcsp} completely classified the complexity of conservative {\vcsp}s.
Since every \minhom\ can be stated as a \vcsp, one might think that the classification provided here is implied by the results in~\cite{kolmogorov:zivny:conservativevcsp}.
This is not the case.
A \vcsp-language is called conservative if it contains all unary $\{0,1\}$-valued cost functions.
The conservative \minhom-languages on the other hand correspond to \vcsp-languages that contain every unary $\{0,\infty\}$-valued cost function.
(Note however that far from all \vcsp-languages that contain every unary $\{0,\infty\}$-valued cost function correspond to a \minhom-language.)

\subsection {\minsol\ on the Three-element Domain}
In this section we fully classify the complexity of \minsol\ on the three-element domain.
\begin{theorem}
\label{thm:minsol}
Let $(\Gamma,\nu)$ be a language over a three-element domain $D$ and $\nu : D \to \qplus$ be injective.
If $(\Gamma,\nu)$ is a min-core and there is no $S \subseteq 2^D$ \st $(\Gamma \cup S,\nu)$ is of type \typegwtp, \typebsm\ or \typegmc, then \minsol$(\Gamma,\nu)$ is \class{NP}-hard.
\end{theorem}

The following two lemmas provide key assistance in the proof of \thmref{thm:minsol}.
The first of the two is a variation of Lemma~3.5 in~\cite{thapper:zivny:fvcsp}.
The lemmas are proved in \sectsref{s:lem:fpol} and~\ref{s:hasconst}.
\begin{lemma}
\label{lem:fpol}
If $\smallcrosspic a b a b \not\in \wclose{\Gamma,\Delta}$, then for every $\sigma \in \eclose{\Gamma,\Delta}$ there is $\omega \in \fpol(\Gamma,\Delta)$ with $f \in \supp(\omega)$ \st $\{f(a,b),f(b,a)\} \ne \{a,b\}$ and $\sigma(f(a,b))+\sigma(f(b,a)) \le \sigma(a)+\sigma(b)$.
\end{lemma}
\begin{lemma}
\label{lem:hasconst}
Let $(\Gamma,\nu)$ be a language over a three-element domain $D$ and $\nu : D \to \qplus$ be injective.
If $(\Gamma,\nu)$ is a min-core and not of type \typegmc, then $\Gamma^c \subseteq \wclose{\Gamma,\nu}$.
\end{lemma}
The proof of \thmref{thm:minsol} contains a somewhat lengthy case-analysis and is deferred to \sectref{s:thm:minsol}.
The case-analysis splits the proof into cases depending on what unary relations that are weighted pp-definable in $(\Gamma,\nu)$.
In each case it is essentially shown that, 
unless a two-element subset $\{x,y\} \subseteq D$ is definable \st \minsol$(\Gamma \cup \wclose{\Gamma,\nu} \cap \binopers,\nu)|_{\{x,y\}}$ is \class{NP}-hard,
in which case also \minsol$(\Gamma,\nu)$ is \class{NP}-hard,
the language $(\Gamma \cup S,\nu)$ is of type \typegmc, \typebsm\ or \typegwtp\ for some $S \subseteq 2^D$.

If $(\Gamma \cup S,\nu)$ is a min-core and of type \typegwtp\ (and not of type \typegmc), then from \lemref{lem:hasconst} it follows that \csp$(\Gamma^c \cup S)$ $\le_p$ \minsol$(\Gamma^c \cup S,\nu)$ $\le_p$ \minsol$(\Gamma \cup S,\nu)$.
Since \minsol$(\Gamma,\nu)$ is a restricted variant of \minhom$(\Gamma,\Delta)$, we therefore, from \thmrefs{thm:wt:p}, \ref{thm:bsm:sl:p} \ref{thm:gmc:p} and~\ref{thm:minsol}, obtain the following.
\begin{theorem}
\label{thm:minsol:2}
Let $(\Gamma,\nu)$ be a language over a three-element domain $D$ and $\nu : D \to \qplus$ be injective.
\minsol$(\Gamma,\nu)$ is in \class{PO} if $(\Gamma,\nu)$ has a min-core $(\Gamma',\nu')$ that is of type \typebsm\ or \typegmc, or if there is $S \subseteq 2^D$ \st \csp$((\Gamma')^c \cup S)$ is in \class P and $(\Gamma' \cup S,\nu')$ is of type \typegwtp\@.
Otherwise \minsol$(\Gamma,\nu)$ is \class{NP}-hard.
\end{theorem}

The following provides an example of use of the classification.
Jonsson, Nordh and Thapper~\cite{maxsolgraph} classified the complexity of \minsol$(\{R\},\nu)$ for all valuations $\nu$ and binary symmetric relations $R$ (\ie graphs) on the three-element domain.
One relation stood out among the others, namely:
$H_5 = \{ (a,c),$ $(c,a),$ $(b,b),$ $(b,c),$ $(c,b),$ $(c,c) \}$,
where $\nu(a)<\nu(b)<\nu(c)$.
If $\nu(a)+\nu(c)<2\nu(b)$ then $\pr_1( \argmin_{(x,y) \in H_5} (\nu(x)+\nu(y)) ) = \{a,c\}$ which means that the relation $\smallmispic c a c a \in \wclose{\{H_5\},\nu}$, and \minsol$(\{H_5\},\nu)$ is \class{NP}-hard by a reduction from the maximum independent set problem.
Otherwise the problem is in \class{PO}.
This was determined in~\cite{maxsolgraph} by linking the problem with, and generalising algorithms for, the critical independent set problem~\cite{cmis:1}.
We note that $\sqcup, \sqcap \in \pol(\{H_5\})$, where $\sqcup, \sqcap$ are commutative idempotent binary operations \st \smash[b]{$\sqcap \in \doa a b \doa c b$, $\sqcup \in \upa a b \upa c b$} and $a \sqcap c = a \sqcup c = b$. This means that $(\{H_5\},\nu)$ is of type \typebsm\@.

\section {Proof of \thmref{thm:wt:p}}
\label{s:proof:1}
Let $I=(V,C,w)$ be an instance of \minhom$(\Gamma, \Delta)$ with measure $m$.
Since \csp$(\Gamma^c)$ is in \class P we can, in polynomial-time, compute the reduced domain $D_v = \{\phi(v) : \phi \in \sol(I)\}$ for every $v \in V$.
Note that $D_v \in \close{\Gamma}$.

Let $f_1,f_2$ be a generalised weak tournament pair on $(A,B)$.
If there for some $v \in V$ is some $\{x,y\} \in  A$ \st $\{x,y\} \subseteq D_v$, then we know that there is $\{a,b\} \in A$ so that $D_v \setminus \{b\} \in \close{\Gamma}$ and $(\Gamma,\Delta)$ admits an \domp a b binary fractional polymorphism $\omega$.
Assume that $\phi_a$ and $\phi_b$ are \st $m(\phi_a) = \min\{m(\phi) : \phi \in \sol(I), \phi(v)=a\}$ and $m(\phi_b) = \min\{m(\phi) : \phi \in \sol(I), \phi(v)=b\}$.

Certainly $g(\phi_a,\phi_b) \in \sol(I)$ for every $g \in \supp(\omega)$.
Because $\omega \in \fpol(\Gamma,\Delta)$ it follows that
\ba
& \ssum{g \in \opers 2} \omega(g) m( g(\phi_a,\phi_b) ) = \ssum{g \in \opers 2} \, \omega(g) \, \ssum{x \in V, \nu \in \Delta} w(x,\nu) \nu( g(\phi_a,\phi_b)(x) ) \\
& \quad =   \ssum{x \in V, \nu \in \Delta} w(x,\nu) \ssum{g \in \opers 2} \omega(g) \nu( g(\phi_a(x),\phi_b(x)) ) \\
& \quad \le \ssum{x \in V, \nu \in \Delta} w(x,\nu) \frac{1}{2} (\nu(\phi_a(x)) + \nu(\phi_b(x))) = \frac{1}{2} ( m(\phi_a) + m(\phi_b) ).
\ea
Since $\omega$ is \domp a b there are functions $\rho, \sigma : \binopers \to \qplus$ \st
$\omega = \rho + \sigma$, $\sum_{g \in \binopers} \rho(g) = \sum_{g \in \binopers} \sigma(g) = \frac{1}{2}$, $g(a,b)=a$ for every $g \in \supp(\rho)$, and $f(a,b)\ne b$ for some $f \in \supp(\sigma)$.
This implies that
\begin{flalign*}
         && \smash{ \frac{1}{2} m( \phi_a ) } + \ssum{g \in \opers 2} \sigma(g) m( g(\phi_a,\phi_b) )  & \le \ssum{g \in \opers 2} \omega(g) m( g(\phi_a,\phi_b) ),&\\
\rlap{so}&& \vphantom{\sum_{i}}\smash{          \ssum{g \in \opers 2} 2\sigma(g) m( g(\phi_a,\phi_b) )}& \le m(\phi_b),                                                 &
\end{flalign*}
which in turn (since $\sum_{g \in \binopers} 2\sigma(g)=1$ and $f(a,b)\ne b$ for some $f \in \supp(\sigma)$) implies that there is $\phi^* \in \sol(I)$ \st $m(\phi^*) \le m(\phi_b)$ and $\phi^*(v) \ne b$.
Hence $b$ can be removed from $D_v$ without increasing the measure of an optimal solution.
To accomplish this the constraint $(v, D_v \setminus \{b\})$ is added.

We repeat this procedure until $\binom{D_v}{2} \cap A = \emptyset$ for every $v \in V$.
Clearly this takes at most $|D|\cdot|V|$ iterations.

Let $f'_1 = f_1[f_1,\overline{f_1}]$ and $f'_2 = f_2[f_2,\overline{f_2}]$.
Note that $f'_1|_{\{x,y\}}$ and $f'_2|_{\{x,y\}}$ are different conservative, idempotent and commutative operations if $\{x,y\} \in B \setminus A$ and projections if $\{x,y\} \in \binom D 2 \setminus B$.
If $f_1|_{\{x,y\}}=\pr_1$ for some $\{x,y\}$, then $f'_1|_{\{x,y\}} = f_1|_{\{x,y\}} = \pr_1$, and if $f_1|_{\{x,y\}}=\pr_2$, then $f'_1|_{\{x,y\}} = \overline{f_1}|_{\{x,y\}} = \overline{\pr_2} = \pr_1$.
So $f'_1|_{\{x,y\}} = \pr_1$ for every $\{x,y\} \in \binom D 2 \setminus B$.
The same arguments apply also for $f'_2$.

Clearly $f'_1|_{D_v}$ and $f'_2|_{D_v}$ are conservative operations for every $v\in V$.
Let $g \in \pol(\Gamma)$ be a ternary idempotent operation that is arithmetical on $\binom D 2 \setminus B$.
Define $g'$ through $g'(x,y,z) = g( f'_1(x, f'_1(y,z)), f'_1(y, f'_1(x,z)), f'_1(z, f'_1(x,y)) )$.
Since $f'_1 = \pr_1$ on $\binom D 2 \setminus B$ also $g'$ is arithmetical on $\binom D 2 \setminus B$.
Since $f'_1$ is conservative, commutative and idempotent on $B \setminus A$ we have $f'_1(x,f'_1(x,y))=f'_1(x,f'_1(y,x))=f'_1(y,f'_1(x,x)) \in \{x,y\}$ for every $\{x,y\} \in B \setminus A$, so $g'$ is conservative on $\binom D 2 \setminus A$.
Note that $f'_1, f'_2, g' \in \pol(\Gamma^+)$ where $\Gamma^+ = \Gamma \cup \{S : S \subseteq D_v \text{ for some } v \in V\})$.
This together with the fact that only a constant number of subsets of $D$ exists means that the modified instance $I$ is easily turned into an instance of the multi-sorted version of \minhom$(\Gamma^+,\nabla_D)$,
where $\nabla_D$ is the set of all functions $D \to \bbbn$, and is solvable in polynomial time~\cite[Theorem 23]{rustem:minsol}.

\section {Proof of \propref{prop:use:dom}}
\label{s:proof:2}
For $\nu \in \Delta$ let $D_{\nu} = \{ x \in D : \nu(x)<\infty\}$.
Let $\Omega = \{f \in \binopers \cap \pol(\Gamma) : \nu(f(x,y))<\infty \text{ for every } \nu \in \Delta \text{ and } x,y \in D_{\nu} \}$,
$\Omega_1 = \{ f \in \Omega: f(a,b)=a \}$, $\Omega_2 = \{ f \in \Omega: f(a,b)=b \}$
and $\Omega_3 = \Omega \setminus (\Omega_1 \cup \Omega_2)$.
The language $(\Gamma,\Delta)$ admits a binary fractional polymorphism that is \domp a b if the following system has a solution $u_g\in \bbbq$, $g\in \Omega$.
\ba
& \ssum{g \in \Omega} u_g \nu( g(x,y) ) \le \frac{1}{2} ( \nu(x) + \nu(y) ) \text{ for } \nu \in \Delta, (x,y) \in D_{\nu}^2, \quad
- u_g \le 0 \text{ for } g \in \Omega,\\ 
& \vphantom{\sum_{i}}\smash{
  \ssum{g \in \Omega} u_g \le  1, \quad 
- \ssum{g \in \Omega} u_g \le -1, \quad
- \ssum{g \in \Omega_1} u_g \le -\frac{1}{2}, \quad \text{and} \quad
  \ssum{g \in \Omega_2} u_g < \frac{1}{2}
  }
\ea
If the system is unsatisfiable, then, by \thmref{thm:motzkin}, there are $v_{\nu,(x,y)}$, $o_g$, $w_1$, $w_2$, $w_3$, $z \in \qplus$ for $\nu \in \Delta, (x,y)\in D_{\nu}^2, g \in \Omega$ \st 
\begin{flalign*}
          && \ssum{\nu \in \Delta, (x,y)\in D_{\nu}^2} \nu(g(x,y)) v_{\nu,(x,y)} - o_g + w_1 - w_2 - w_3 &= 0, &\mathllap{g \in \Omega_1,}\\
          && \ssum{\nu \in \Delta, (x,y)\in D_{\nu}^2} \nu(g(x,y)) v_{\nu,(x,y)} - o_g + w_1 - w_2 + z   &= 0, &\mathllap{g \in \Omega_2,}\\
          && \ssum{\nu \in \Delta, (x,y)\in D_{\nu}^2} \nu(g(x,y)) v_{\nu,(x,y)} - o_g + w_1 - w_2       &= 0, &\mathllap{g \in \Omega_3,}\\
\rlap{and}&& \vphantom{\sum_{i}}\smash{ \ssum{\nu \in \Delta, (x,y)\in D_{\nu}^2} \frac{1}{2} (\nu(x)+\nu(y)) v_{\nu,(x,y)} + w_1 - w_2 - \frac{1}{2} w_3 + \frac{1}{2} z} &= \alpha, &
\end{flalign*}
where either $\alpha < 0$ or $\alpha = 0$ and $z>0$.
Hence, for every $g \in \Omega_1$ and $h \in \Omega_2$,
\ba
& \vphantom{\sum_{i}}\smash{ \ \ \ \ \ \ssum{\nu \in \Delta, (x,y)\in D_{\nu}^2} (\nu(x)+\nu(y)) v_{\nu,(x,y)} +  o_g + o_h = \ssum{\nu \in \Delta, (x,y)\in D_{\nu}^2} (\nu(g(x,y))+\nu(h(x,y))) v_{(x,y),\nu} + \alpha.}
\ea
Note that since $\pr_1 \in \Omega_1$ and $\pr_2 \in \Omega_2$ we must have $\alpha=0$, $o_{\pr_1} = o_{\pr_2} = 0$, and $z>0$.
This means that
\ba
& \min_{g \in \Omega_1} \ssum{\nu \in \Delta, (x,y)\in D_{\nu}^2} \nu(g(x,y)) v_{\nu,(x,y)}
= \ssum{\nu \in \Delta, (x,y)\in D_{\nu}^2} \nu(\pr_1(x,y)) v_{\nu,(x,y)} = - w_1 + w_2 + w_3 \\
& \vphantom{\sum_{i}}\smash{ \quad > - w_1 + w_2 - z = \ssum{\nu \in \Delta, (x,y)\in D_{\nu}^2} \nu(\pr_2(x,y)) v_{\nu,(x,y)}
= \min_{g \in \Omega_2} \ssum{\nu \in \Delta, (x,y)\in D_{\nu}^2} \nu(g(x,y)) v_{\nu,(x,y)}. }
\ea

Create an instance $I$ of \minhom$(\Gamma,\Delta)$ with variables $D^2$, and objective 
\ba
m(\phi) &=  \ssum{\nu \in \Delta, (x,y) \in D_{\nu}^2} v_{\nu,(x,y)} \nu(\phi(x,y))
 + \epsilon \ssum{\nu \in \Delta, (x,y) \in D_{\nu}^2} \nu(\phi(x,y)),
\ea
where $\epsilon > 0$ is choosen small enough so that $\phi \in \argmin_{\phi' \in \Omega_1} m(\phi')$ implies
$\phi \in \argmin_{\phi' \in \Omega_1} \sum_{\nu \in \Delta, (x,y) \in D_{\nu}^2} v_{\nu,(x,y)} \nu(\phi(x,y))$.
Such a number $\epsilon$ can always be found.
Note that a solution $\phi$ to $I$ with finite measure is a function $D^2 \to D$ \st $\nu(\phi(x,y))<\infty$ for every $\nu \in \Delta$ and $(x,y) \in D_{\nu}^2$.

Pick, for every $g \in \binopers \setminus \pol(\Gamma)$, a relation $R_g \in \Gamma$ \st $g$ does not preserve $R_g$.
Add for each pair of tuples $t^1,t^2 \in R_g$ the constraint $(( (t^1_1,t^2_1), \dots, (t^1_{\ar(R_g)},t^2_{\ar(R_g)})), R_g)$.
This construction is essentially the second order indicator problem~\cite{indicator}.
Now a solution to $I$ is a binary polymorphism of $\Gamma$.
Hence, if $\phi$ is a solution to $I$ with finite measure, then $\phi \in \Omega$.
Clearly $\pr_1$ and $\pr_2$ satisfies all constraints and are solutions to $I$ with finite measures.
Let $\nu(x) = \min_{g \in \sol(I) : g(a,b)=x} m(g)$.
Note that $\nu \in \eclose{\Gamma,\Delta}$ and $\infty > \nu(a) > \nu(b)$.
This completes the proof.

\section {Proof of \thmref{thm:conservative}}
\label{s:proof:3}
The proof follows the basic structure of the arguments given in~\cite{rustem:minhom}.
A key ingredient of our proof will be the use of \thmref{thm:wt:p} and \propref{prop:use:dom}.

Let $\Gamma^+ = \Gamma \cup (2^D \cup 2^{D^2}) \cap \wclose{\Gamma,\Delta}$.
Note that if $(\Gamma^+, \Delta)$ is of type \typegwtp, then so is also $(\Gamma \cup S,\Delta)$, for some $S \subseteq 2^D$.
Since \minhom$(\Gamma^+, \Delta)$ is polynomial-time reducible to \minhom$(\Gamma,\Delta)$ we therefore assume that $\Gamma^+ \subseteq \Gamma$.
We also assume $\Gamma^c \subseteq \Gamma$.
Obviously \csp$(\Gamma)$ is polynomial-time reducible to \minhom$(\Gamma,\Delta)$.
In what follows we therefore assume that \csp$(\Gamma)$ is in \class{P}.

Let $B \subseteq \binom D 2$ be a minimal set \st
all binary operations in $\pol(\Gamma)$ are projections on $\binom D 2 \setminus B$
and for every $\{a,b\} \in \binom D 2 \setminus B$ there is a ternary operation in $\pol(\Gamma)$ that is arithmetical on $\{\{a,b\}\}$.
Then, let $A$ be a maximal subset of $B$ \st for every $\{a,b\} \in A$ there is $\omega \in \fpol(\Gamma,\Delta)$ \st $\omega$ is either \domp a b or \domp b a.
Let $T$ be the undirected graph $(M,P)$, where
$M = \{ (a,b) : \{a,b\} \in \Gamma \cap B \setminus A \}$ and
$P = \{ ((a,b), (c,d)) \in M^2 : \pol(\Gamma) \cap \doa a b \doa c d = \emptyset \}$.

By \propref{prop:use:dom} we know that for every $(a,b) \in M$, there are $\nu,\tau \in \eclose{\Gamma,\Delta}$ \st $\nu(b)<\nu(a)<\infty$ and $\tau(a)<\tau(b)<\infty$.
By the classification of \minhom\ on two-element domains, see \eg~\cite[Theorem~3.1]{rustem:minhom}, and by the fact that
if $f,m \in \pol(\Gamma)$ are idempotent, $m$ is arithmetical on $\{\{x,y\}\}$ and \smash[b]{$f \in \upa x y$}, then $m'(u,v) = m(u,f(u,v),v)$ satisfies \smash[b]{$m' \in \doa x y$},
we have the following.
\begin{lemma}
\label{lem:get:hardness}
Either; for every $(a,b) \in M$ there are binary operations $f,g \in \pol(\Gamma)$ \st $f|_{\{a,b\}}$ and $g|_{\{a,b\}}$ are two different idempotent, conservative and commutative operations, or \minhom$(\Gamma,\Delta)$ is \class{NP}-hard.
\end{lemma}
\begin{lemma}[{\cite[Theorem~5.3]{rustem:minhom}}]
\label{lem:bipartite}
If $T$ is bipartite, then there are binary operations $f,g \in \pol(\Gamma)$ \st for every $(a,b) \in M$,
$f|_{\{a,b\}}$ and $g|_{\{a,b\}}$ are different idempotent conservative and commutative operations, or \minhom$(\Gamma,\Delta)$ is \class{NP}-hard.
\end{lemma}
\begin{lemma}[{\cite[Theorem~5.4]{rustem:minhom}}]
\label{lem:arithmetical}
Let $C \subseteq \binom D 2$. If $C \subseteq \Gamma$ and for each $\{a,b\} \in C$ there is a ternary operation $m^{\{a,b\}} \in \pol(\Gamma)$ that is arithmetical on $\{\{a,b\}\}$, then there is $m \in \pol(\Gamma)$ that is arithmetical on $C$.
\end{lemma}
So, if $T$ is bipartite and $(\Gamma,\Delta)$ is conservative, there is a generalised weak tournament pair on $(A,B)$ and an arithmetical polymorphism on $\binom D 2 \setminus B$.
Here $(\Gamma, \Delta)$ is of type \typegwtp, and by \thmref{thm:wt:p}, we can conclude that \minhom$(\Gamma,\Delta)$ is polynomial-time solvable.

This following lemma finishes the proof of \thmref{thm:conservative}.
A corresponding result, for the case when $\Delta$ is the set of all functions $D \to \bbbn$, is also achieved in~\cite{rustem:minhom}.
Our proof strategy is somewhat different from that in~\cite{rustem:minhom}, though.
\begin{lemma}
\label{lem:notbipartite}
If $T$ is not bipartite, then \minhom$(\Gamma,\Delta)$ is
\class{NP}-hard.
\end{lemma}
\begin{proof}
We will show that if $T$ is not bipartite, then $\smallcrosspic b a b a \in \close{\Gamma}$ for some $(a,b) \in M$.
From this it follows, using \lemref{lem:get:hardness}, that \minhom$(\Gamma,\Delta)$ is \class{NP}-hard.
We will make use of the following result.
\begin{lemma}[{\cite[Lemma~4.2]{rustem:minhom}}]
\label{lem:edgerels}
If $((a,b),(c,d)) \in P$, then either $\smallcrosspic a b c d \in \close{\Gamma}$ or $\smallmispic a b c d \in \close{\Gamma}$.
\end{lemma}
Since $\Gamma^+ \subseteq \Gamma$, and since there are functions $\nu,\tau \in \eclose{\Gamma,\Delta}$ \st $\nu(b)<\nu(a)<\infty$ and $\tau(d)<\tau(c)<\infty$, we immediately get the following.
\begin{corollary}
\label{cor:edgerels:simp}
If $((a,b),(c,d)) \in P$, then $\smallcrosspic a b c d \in \Gamma$.
\end{corollary}
Since $T$ is not bipartite it must contain an odd cycle $(a_0,b_0),$ $(a_1,b_1),$ $\dots,$ $(a_{2k},b_{2k}),$ $(a_0, b_0)$.
This means, according to \corref{cor:edgerels:simp}, that $\Gamma$ contains relations $\rho_{0,1},$ $\rho_{1,2},$ $\dots,$ $\rho_{2k-1,2k},$ $\rho_{2k,0}$ where $\rho_{i,j} = \smallcrosspic{a_i}{b_i}{a_j}{b_j}$.
Since the cycle is odd this means that
$\rho_{0,1} \circ \rho_{1,2} \circ \dots \circ \rho_{2k-1,2k} \circ \rho_{2k,0} = \smallcrosspic{a_0}{b_0}{a_0}{b_0} \in \close{\Gamma}$.
\hfill\qed
\end{proof}

\section {Concluding Remarks}
We have fully classified the complexity of \minsol\ on domains that contain at most three elements and the complexity of conservative \minhom\ on arbitrary finite domains.

Unlike for \csp\ there is no widely accepted conjecture for the complexity of \vcsp\@.
This makes the study of small-domain {\vcsp}s an exciting and important task.
We believe that a promising approach for this project is to study \minhom\ --- it is interesting
for its own sake and likely easier to analyse than the general \vcsp\@.

A natural continuation of the work presented in this paper would be to classify \minhom\ on domains of size three.
This probably is a result within reach using known techniques.
Another interesting question is what the complexity of three-element \minsol\ is when the domain valuation is not injective
(we note that if the valuation is constant the problem collapses to a \csp\ whose complexity has been classified by Bulatov~\cite{bulatov}, but situations where \eg $\nu(a)=\nu(b)<\nu(c)$ are not yet understood).

\bigskip
\noindent{\bf Acknowledgements.}
I am thankful to Peter Jonsson for rewarding discussions and to Magnus Wahlstr\"{o}m for helpful comments regarding the presentation of the results.
I am also grateful to three anonymous reviewers for their useful feedback.

\newpage
\appendix

\section {Proof of \thmref{thm:minsol}}
\label{s:thm:minsol}
We will in this section use $x \domps\omega y$ to denote that the binary fractional polymorphism $\omega$ is \domp x y.
For binary operations $f$ and $g$ the notation $\{ f \mapsto \frac 1 2, g \mapsto \frac 1 2 \}$ is used for the fractional polymorphism mapping $f$ and $g$ to $\frac 1 2$, and all other binary operations to $0$.
We will need the following lemma.
\begin{lemma}
\label{lem:extend:arithm}
Let $\Gamma$ be a set of finitary relations on $D = \{a,b,c\}$.
If there exists a binary operation $f \in \pol(\Gamma)$ and a ternary operation $m \in \pol(\Gamma)$ \st
$f$ and $m$ are idempotent,
$f|_{\{a,b\}}$ and $f|_{\{b,c\}}$ are projections,
$f(a,c)=f(c,a)=b$, and $m$ is arithmetical on $\binom D 2 \setminus \{ \{a,c\} \}$,
then there is an idempotent ternary operation $m' \in \pol(\Gamma)$ \st $m'$ is arithmetical on $\binom D 2$.
\end{lemma}
\begin{proof}
Assume \wlg that $f|_{\{a,b\}} = f|_{\{b,c\}} = \pr_1$.
If this does not hold, then $f' = f[f,\overline{f}]$ is another polymorphism that does satisfy the condition.

Let $g(x,y,z) = m( f(m(x,y,z),z), f(m(y,x,z),z), z )$.
It can be checked that $g$ is arithmetical on $\binom D 2 \setminus \{ \{a,c\} \}$ and additionally satisfies $g(a,a,c)=c$ and $g(c,c,a)=a$.
Define $h(x,y,z) = g(z, f(y,z), g(x,f(x,y),f(x,z)))$ and $m'(x,y,z) = g( f(x,y), f(y,x), h(x,y,z) )$.
It is straightforward to verify that $m'$ is indeed arithmetical on $\binom D 2$.
\hfill\qed
\end{proof}

Let $(\Gamma,\nu)$ be a min-core language on the domain $D = \{a,b,c\}$
and $\nu(a) < \nu(b) < \nu(c) < \infty$.
Let $\Gamma^+ = \Gamma \cup (2^D \cup 2^{D^2}) \cap \wclose{\Gamma,\nu}$.
Note that if $(\Gamma^+, \nu)$ 
is of type \typegwtp\ (\typegmc, \typebsm), 
then there is $S \subseteq 2^D$ \st also $(\Gamma \cup S,\nu)$ is of type \typegwtp\ (\typegmc, \typebsm).
Since \minhom$(\Gamma^+, \nu)$ is polynomial-time reducible to \minhom$(\Gamma,\nu)$ we therefore assume that $\Gamma^+ \subseteq \Gamma$.
By \lemref{lem:hasconst} we can assume that $\Gamma^c \subseteq \Gamma$ since otherwise $\Gamma$ is of type \typegmc\@.

In the following we will assume that \minhom$(\Gamma,\nu)$ is not \class{NP}-hard and show that this implies that $(\Gamma,\nu)$ is of type  \typegwtp, \typegmc\ or \typebsm\@.

Let $B \subseteq \binom D 2$ be a minimal set \st for every $\{a,b\} \in \binom D 2 \setminus B$ there is a ternary operation in $\pol(\Gamma)$ that is arithmetical on $\{a,b\}$ and all binary operations in $\pol(\Gamma)$ are projections on $\{a,b\}$.
Then, let $A$ be a maximal subset of $B$ \st for every $\{a,b\} \in A$ there is $\omega \in \fpol(\Gamma,\nu)$ \st $\omega$ is either \domp a b or \domp b a.
We can assume that there are $f_1,f_2,m \in \pol(\Gamma)$ \st for every $\{x,y\} \in \Gamma \cap B \setminus A$ it holds that $f_1|_{\{x,y\}}$ and $f_2|_{\{x,y\}}$ are different idempotent, commutative and conservative operations, and on $\binom D 2 \setminus B$, $m$ is arithmetical while $f,g$ are projections. Otherwise, by Lemmas~\ref{lem:get:hardness}, \ref{lem:bipartite}, \ref{lem:arithmetical} and~\ref{lem:notbipartite}, \minhom$(\Gamma,\nu)$ is \class{NP}-hard.
This means, unless \minhom$(\Gamma,\nu)$ is \class{NP}-hard, that if there is any $\{x,y\} \subseteq D$ \st $(\Gamma,\nu)$ admits
a fractional polymorphism that is \domp x y and $D \setminus \{y\} \in \Gamma$,
or if $\Gamma$ is conservative, then $(\Gamma,\nu)$ is of type \typegwtp\@.
In the following we therefore assume that this is not the case.

We split the rest of the proof in seven cases depending on which unary relations $\Gamma$ contains.

\subsection{$\{a,b\} \not\in \Gamma$, $\{a,c\} \in \Gamma$, $\{b,c\} \in \Gamma$}
By \lemref{lem:fpol}, there is $\omega \in \fpol(\Gamma,\nu)$ \st some $f \in \supp(\omega)$ satisfies $f \in \doa b a$.
This means, unless $a \domps\omega b$, that every $g \in \supp(\omega)$ is conservative and that there is \smash[b]{$f' \in \supp(\omega) \cap \upa b a$}.
\be
\item
If $\smallcrosspic c a c a \not\in \Gamma$ and $\smallcrosspic c b c b \not\in \Gamma$, then by \lemref{lem:fpol} we may assume that there is $g,h \in \supp(\omega)$ \st $g \in \doa c a$ and $h \in \doa c b$.
If there is any $i \in \supp(\omega)$ \st $i \in \doa b a$ and \smash[t]{$i \not\in \upa c a$, then $i' = g[i,\overline{i}] \in \doa b a \doa c a$}.
In this case
$\psi = \{ i' \mapsto \frac 1 4, \overline{i'} \mapsto \frac 1 4, h \mapsto \frac 1 4, \overline{h} \mapsto \frac 1 4 \} \in \fpol(\Gamma,\nu)$,
so unless $a \domps\psi b$ we have $h \in \upa b a$.
Note that $h,i' \in \pol(\Gamma \cup \{ \{a,b\} \})$ 
and unless $a \domps\psi c$ or $b \domps\psi c$ it must hold that $i',h$ are complementary and $(\Gamma \cup \{ \{a,b\} \},\nu)$ is of type \typegwtp\@.

Otherwise it follows by symmetry that $\supp(\omega) \cap \doa b a \subseteq \upa c a \upa c b$, $\supp(\omega) \cap \doa c a \subseteq \upa b a \upa c b$
and \smash[t]{$\supp(\omega) \cap \doa c b \subseteq \upa b a \upa c a$}, this contradicts that $\omega \in \fpol(\Gamma,\nu)$.

\item
If $\smallcrosspic c a c a \in \Gamma$ and $\smallcrosspic c b c b \not\in \Gamma$, then by \lemref{lem:fpol} we can assume \smash[t]{$g \in \doa c b$} for some $g \in \supp(\omega)$.
Here \smash{$\psi = \{ g \mapsto \frac 1 4, \overline{g} \mapsto \frac 1 4, f \mapsto \frac 1 4, \overline{f} \mapsto \frac 1 4 \} \in \fpol(\Gamma,\nu)$}.
Unless $a \domps\psi b$ we have $g \in \upa b a$.
Note that $f,g \in \pol(\Gamma \cup \{ \{a,b\} \})$ and unless $b \domps\psi c$
it must hold that $g, f$ are complementary, so $(\Gamma \cup \{ \{a,b\},\nu)$ is of type \typegwtp\@.

\item
If $\smallcrosspic c a c a \not\in \Gamma$ and $\smallcrosspic c b c b \in \Gamma$, then by symmetry to the case above, $(\Gamma,\nu)$ is of type \typegwtp\@.

\item
If $\smallcrosspic c a c a \in \Gamma$ and $\smallcrosspic c b c b \in \Gamma$, then by $f,f'$, $(\Gamma,\nu)$ is of type \typegwtp\@.
\ee

\subsection{$\{a,b\} \in \Gamma$, $\{a,c\} \in \Gamma$, $\{b,c\} \not\in \Gamma$}
By \lemref{lem:fpol} there is $\omega \in \fpol(\Gamma,\nu)$ \st $(f,\overline f)(b,c) \not\in \smallmispic c b c b$ for some $f \in \supp(\omega)$.
\be
\item
If $\smallcrosspic b a b a \not\in \Gamma$ and $\smallcrosspic c a c a \not\in \Gamma$,
then, by \lemref{lem:fpol}, we may assume \smash{$g \in \doa b a$, $h \in \doa c a$} for some $g,h \in \supp(\omega)$.

\be
\item
If there is $i \in \pol(\Gamma)$ \st \smash[t]{$i \in \doa b a \doa c a$},
then \smash[t]{$\psi = \{ i \mapsto \frac 1 2, \overline{i} \mapsto \frac 1 2 \} \in \fpol(\Gamma,\nu)$} and $a \domps\psi b$, or \smash[t]{$i \in \upa c b$}.
Note that $f' = i[f,\overline{f}]$ is commutative and satisfies $f'(b,c) \in \{a,b\}$.
So $\psi = \{ f' \mapsto \frac 1 2, i \mapsto \frac 1 2 \} \in \fpol(\Gamma,\nu)$
and $f', i$ must be complementary  unless $a \domps\psi b$, $a \domps\psi c$ or $c \domps\psi b$.
This means that $(\Gamma,\nu)$ is of type \typegwtp\@.

\item
Otherwise $\pol(\Gamma) \cap \doa b a \subseteq \upa c a$ and $\pol(\Gamma) \cap \doa c a \subseteq \upa b a$.
\bi
\item
If there is $r \in \pol(\Gamma)$ \st $(r,\overline r)(b,c) \not\in \smallmispic c b c b$ and $r$ is a projection on both $\{a,b\}$ and $\{a,c\}$, then
\bi
\item
if $(r,\overline r)(b,c) \in \smallcrosspic c a c a$ we have $\psi = \{ r \mapsto \frac 1 2, \overline{r} \mapsto \frac 1 2 \} \in \fpol(\Gamma,\nu)$ and $c \domps\psi b$, 
\item
if $(r,\overline r)(b,c) = \smallmispic b a b a$, then
$\psi = \{ g[r,\overline{r}] \mapsto \frac 1 2, h[r,\overline{r}] \mapsto \frac 1 2 \} \in \fpol(\Gamma,\nu)$ and $b \domps\psi c$, 
\item
otherwise $(r,\overline r)(b,c) = (a,a)$.
Let $g' = g[r,\overline{r}]$, $h' = h[r,\overline{r}]$ and $f' = g'[\pr_1,h']$.
Is is easy to check that
$\psi = \{ f' \mapsto \frac 1 2, \overline{f'} \mapsto \frac 1 2 \} \in \fpol(\Gamma,\nu)$ and $c \domps\psi b$.
\ei

\item
If there is $r \in \pol(\Gamma)$ \st $(r,\overline r)(b,c) \not\in \smallmispic c b c b$ and \smash[b]{$r \in \doa b a \upa c a$}:
\bi
\item
If $(r,\overline r)(b,c) \in \smallflipmispic c a c a$, consider the following.
\bi
\item
If $(h,\overline h)(b,c) \in \smallmispic c a c a$, let $r' = h[r,\overline{r}]$ and $h' = g[h,\overline{h}]$.
Now $\psi = \{ r' \mapsto \frac 1 2, h' \mapsto \frac 1 2 \} \in \fpol(\Gamma,\nu)$ and $c \domps\psi b$.
\item
If $(h,\overline h)(b,c) \in \smallmispic b a b a$, let $r' = h[r,\overline{r}]$ and $h' = h[h,\overline{h}]$.
Now $\psi = \{ r' \mapsto \frac 1 2, h' \mapsto \frac 1 2 \} \in \fpol(\Gamma,\nu)$ and $b \domps\psi c$.
\item
If $(h,\overline h)(b,c) = (a,a)$, let $h' = h[\pr_1,h]$.
Here \smash{$h' \in \upa b a \doa c a$} and $(h',\overline{h'})(b,c) \in \smallmispic b a b a$, so the previous case applies.
\item
If $(h,\overline h)(b,c) \in \smallcrosspic c b c b$, let $h' = r[h,\overline{h}]$.
Here \smash[b]{$h' \in \upa b a \doa c a$} and also $(h',\overline{h'})(b,c) \in \smallflipmispic c a c a$, so other cases can be used.
\ei
\item
If $(r,\overline r)(b,c) \in \smallmispic b a b a$, let $r' = h[r,\overline{r}]$.
Here either $r', h'=r'[h,\overline{h}]$ are complementary and $(\Gamma,\nu)$ is of type \typegwtp,
or $\psi = \{ r' \mapsto \frac 1 2, h' \mapsto \frac 1 2 \} \in \fpol(\Gamma,\nu)$ and $b \domps\psi c$.
\ei

\item
If there is $r \in \pol(\Gamma)$ \st $(r,\overline r)(b,c) \not\in \smallmispic c b c b$ and \smash[t]{$r \in \upa b a \doa c a$},
then by arguments symmetric to the ones above, $(\Gamma,\nu)$ is of type \typegwtp\@.

\item
Otherwise, every $r \in \pol(\Gamma)$ \st $(r,\overline r)(b,c) \not\in \smallmispic c b c b$ satisfies \smash{$r \in \upa b a \upa c a$}.
This contradicts that $\omega \in \fpol(\Gamma,\nu)$.
\ei
\ee

\item
If $\smallcrosspic b a b a \not\in \Gamma$ and $\smallcrosspic c a c a \in \Gamma$,
then, by \lemref{lem:fpol}, we may assume \smash{$h \in \doa b a$} for some $h \in \supp(\omega)$.

If there is $f' \in \supp(\omega)$ \st $(f',\overline{f'})(a,b) \in \smallcrosspic b a b a$ and $(f',\overline{f'})(b,c) \not\in \smallmispic c b c b$, then $\psi = \{h[f',\overline{f'}] \mapsto \frac 1 2, \overline{h[f',\overline{f'}]} \mapsto \frac 1 2\} \in \fpol(\Gamma,\nu)$ and $a \domps\psi b$.
And, if there is $i \in \supp(\omega)$ \st $i \in \doa b a$ and $i \not\in \upa c b$, then $\psi = \{i \mapsto \frac 1 2, \overline{i} \mapsto \frac 1 2\} \in \fpol(\Gamma,\nu)$ and $a \domps\psi b$.
Since we can assume that no operation in $\supp(\omega)$ is a projection on both $\{a,b\}$ and $\{b,c\}$,
this means that if $j \in \supp(\omega)$ and $j \not\in \upa b a$, then $j \in \upa c b$.
So, unless $\omega$ is \st $a \domps\omega b$ or $c \domps\omega b$, it must hold that $f,h$ are complementary, and $(\Gamma,\nu)$ is of type \typegwtp\@.

\item
If $\smallcrosspic b a b a \in \Gamma$ and $\smallcrosspic c a c a \not\in \Gamma$,
then, arguments symmetric to those in the case above establishes that $(\Gamma,\nu)$ is of type \typegwtp\@.

\item
If $\smallcrosspic b a b a \in \Gamma$ and $\smallcrosspic c a c a \in \Gamma$, the following holds.
If $(f,\overline f)(b,c) \ne (a,a)$, then $\psi = \{f \mapsto \frac 1 2, \overline f \mapsto \frac 1 2\} \in \fpol(\Gamma,\nu)$ and either $b \domps\psi c$ or $c \domps\psi b$.
Otherwise $(f,\overline f)(b,c) = (a,a)$ and, by \lemref{lem:extend:arithm}, $(\Gamma,\nu)$ is of type \typegwtp\@.
\ee

\subsection{$\{a,b\} \in \Gamma$, $\{a,c\} \not\in \Gamma$, $\{b,c\} \in \Gamma$}

By \lemref{lem:fpol} there is $\omega \in \fpol(\Gamma,\nu)$ and $f \in \supp(\omega)$ \st \smash{$(f,\overline f)(a,c) \in \smallallpic b a b a$}
and $\nu(f(a,c)) + \nu(f(c,a)) \le \nu(a)+\nu(c)$.
\be
\item
If \smash{$\smallcrosspic b a b a \not\in \Gamma$ and $\smallcrosspic c b c b \not\in \Gamma$},
then, by \lemref{lem:fpol}, we may assume that there is $g,h \in \supp(\omega)$ \st $g \in \doa b a$ and $h \in \doa c b$.
Since $\{a,c\} \not\in \close{\Gamma}$ there is $f \in \pol(\Gamma)$ \st $f(c,a)=b$.
We can assume $f(a,c) \in \{a,b\}$ since otherwise $f'=h[f,\overline f]$ satisfies the property.

If there is $i \in \pol(\Gamma,\nu) \cap \doa b a \doa c b$, then $h' = i[ i[\pr_1,f], \overline{i[\pr_1,f]} ] \in \doa b a \doa c a \doa c b$.
Here $\psi = \{ h' \mapsto \frac 1 2, \overline{ h' } \mapsto \frac 1 2 \} \in \fpol(\Gamma,\nu)$ and $a \domps\psi c$.

Otherwise $\pol(\Gamma) \cap \doa c b \subseteq \upa b a$ and $\pol(\Gamma) \cap  \doa b a \subseteq \upa c b$.
In the following we assume $f(a,c)=f(c,a)=b$ since if this does not hold, then $g[f,\overline{f}]$ satisfies the property.
\bi
\item
If thers is $f \in \supp(\omega)$ that is a projection on $\{a,b\}$, $\{b,c\}$ and satisfies $(f,\overline f)(a,c) \in \smallallpic b a b a$, then
\bi
\item
if $(f,\overline f)(a,c) \in \smallflipmispic b a b a$,
then $\psi = \{ f \mapsto \frac 1 2, \overline{f} \mapsto \frac 1 2 \} \in \fpol(\Gamma,\nu)$ and $a \domps\psi c$,
\item
otherwise $(f,\overline f)(a,c) = (b,b)$.
Note that we might assume $2 \nu(b) \le \nu(a)+\nu(c)$, since otherwise we can modify $\omega$ by setting $\omega(f)=0$ and rescaling the function so that $\sum_{f \in \binopers} \omega(f)=1$. This gives another fractional polymorphism that satisfies our conditions.
Since $h'= h[f,\overline{f}] \in \doa a b \doa c b$, $g' = g[f,\overline{f}] \in \upa a b \upa c b$ and $h'(a,c)=h'(c,a)=g'(a,c)=g'(c,a)=b$,
this means that $(\Gamma,\nu)$ is of type \typebsm\@.
\ei

\item
Otherwise we can assume that $\supp(\omega) = \Omega_1 \cup \Omega_2$, where $\Omega_1 = \supp(\omega) \cap \doa b a \upa c b$
and $\Omega_2 = \supp(\omega) \cap \upa b a \doa c b$.
\bi
\item
If there is $f \in \supp(\omega)$ \st \smash{$(f,\overline f)(a,c) \in \smallflipmispic b a b a$}, then \smash{$f'=g[f,\overline{f}] \in \doa c a$}
and $f \in \doa b a \upa c b$ or $f \in \upa b a \doa c b$.
Assume \wlg the latter holds, then $f',g'=f'[g,\overline{g}]$ are complementary and $(\Gamma,\nu)$ is of type \typegwtp,
or $\psi = \{f' \mapsto \frac 1 2, g' \mapsto \frac 1 2\} \in \fpol(\Gamma,\nu)$ and $a \domps\psi c$.

\item
Otherwise every $f \in \supp(\omega)$ \st $(f,\overline f)(a,c) \in \smallallpic b a b a$ is commutative and satisfies $f(a,c)=b$.
If there is any $f \in \supp(\omega)$ \st $(f,\overline f)(a,c) \in \smallcrosspic c b c b$, then $\omega$ can be modified by changing $f$ to $h[f,\overline{f}]$.
We therefore assume that every $f \in \supp(\omega)$ that is not a projection on $\{a,c\}$ is commutative on $\{a,c\}$ and satisfies $f(a,c) \in \{b,c\}$.
\bi
\item
If $\nu(a)+\nu(c) < 2\,\nu(b)$, then since not all operations in $\supp(\omega)$ are projection on $\{a,c\}$ it is impossible that $\omega \in \fpol(\Gamma,\nu)$.
\item
If $\nu(a)+\nu(c) \ge 2\,\nu(b)$, then since $\omega \in \fpol(\Gamma,\nu)$ and $\nu(a) < \nu(b)$ there must be $f_1 \in \Omega_1, f_2 \in \Omega_2$ \st $f_1(a,c)=f_2(a,c)=b$.
So $(\Gamma,\nu)$ is in this case of type \typebsm\@.
\ei
\ei
\ei

\item
If $\smallcrosspic b a b a \not\in \Gamma$ and $\smallcrosspic c b c b \in \Gamma$,
then, by \lemref{lem:fpol}, we may assume \smash{$g \in \doa b a$} for some $g \in \supp(\omega)$.
Assume \wlg that $g|_{\{b,c\}} = \pr_1$.
\bi
\item
If \smash[t]{$(f,\overline f)(a,c) \in \smallflipmispic b a b a$}, set $f'=g[f,\overline f]$.
Note that $\psi = \{f' \mapsto \frac 1 4, \overline{f'} \mapsto \frac 1 4, g \mapsto \frac 1 4, \overline{g} \mapsto \frac 1 4\} \in \fpol(\Gamma,\nu)$.
Unless $a \domps\psi c$ it holds that $g \in \upa c a$.
This means that $f',g \in \pol(\Gamma \cup \{ \{a,c\} \})$, so either $a \domps\psi b$
or $f',g$ are complementary and $(\Gamma \cup \{ \{a,c\} \},\nu)$ is of type \typegwtp\@.
\item
Otherwise $f(a,c)=f(c,a)=b$, and
\smash{$g'=g[ g[f,\pr_1], \overline{g[f,\pr_1]} ] \in \doa c a \doa b a$},
so $\psi = \{ g' \mapsto \frac 1 2, \overline{g'} \mapsto \frac 1 2\} \in \fpol(\Gamma,\nu)$ and $a \domps\psi c$.
\ei

\item
If $\smallcrosspic b a b a \in \Gamma$ and $\smallcrosspic c b c b \not\in \Gamma$,
then, by \lemref{lem:fpol}, we may assume \smash{$h \in \doa c b$} for some $h \in \supp(\omega)$.
Assume \wlg that $h|_{\{a,b\}}=\pr_1$.

Since $\{a,c\} \not\in \close{\Gamma}$ there is a binary operation $f \in \pol(\Gamma)$ \st $f(a,c)=b$.
Let $f' = h[f,\overline{f}]$ and note that $(f',\overline{f'})(a,c) \in \smallmispic b a b a$.
\bi
\item 
If $(h,\overline h)(a,c) \in \smallflipmispic c a c a \cup \smallcrosspic b a b a$,
then $\psi = \{h \mapsto \frac 1 2, \overline{h} \mapsto \frac 1 2\} \in \fpol(\Gamma,\nu)$ and $b \domps\psi c$.
\item
If $(h,\overline h)(a,c) = (b,b)$, then let $h' = h[\pr_1,h]$ and note that
$\psi = \{h' \mapsto \frac 1 2, \overline{h'} \mapsto \frac 1 2\} \in \fpol(\Gamma,\nu)$ and $b \domps\psi c$.
\item 
If $(h,\overline h)(a,c) \in \smallcrosspic c b c b$, then $h'=h[h,\overline{h}] \in \doa c b$, $h'|_{\{a,b\}} = \pr_1$
and $(h,\overline h)(a,c) = (b,b)$, so the previous case is applicable.
\item
If $(h,\overline h)(a,c) = (c,c)$, then
$h'=h[f',h] \in \doa c b$, $(h',\overline{h'})(a,c) \in \smallflipmispic c b c b$ and $h'$ is a projection.
So one of the previous cases apply.
\ei

\item
If $\smallcrosspic b a b a \in \Gamma$ and $\smallcrosspic c b c b \in \Gamma$, then,
\bi
\item
if $(f,\overline f)(a,c) \in \smallflipmispic b a b a$,
then $\psi = \{f \mapsto \frac 1 2, \overline f \mapsto \frac 1 2\} \in \fpol(\Gamma,\nu)$ and $a \domps\psi c$,
\item
otherwise $(f,\overline{f})(a,c)=(b,b)$ and, by \lemref{lem:extend:arithm}, $(\Gamma,\nu)$ is of type \typegwtp\@.
\ei
\ee

\subsection{$\{a,b\} \in \Gamma$, $\{a,c\} \not\in \Gamma$, $\{b,c\} \not\in \Gamma$}
By \lemref{lem:fpol} there is $\omega \in \fpol(\Gamma,\nu)$ and $f \in \supp(\omega)$ \st $(f,\overline f)(a,c) \in \smallallpic b a b a$.
By \lemref{lem:hascross}, $\smallcrosspic b a b a \in \Gamma$, or $(\Gamma,\nu)$ is of type \typegmc\@.
Since $\{b,c\} \not\in \close{\Gamma}$ we have $g \in \pol(\Gamma)$ \st
\smash[b]{$(g,\overline g)(b,c) \in \mypic{ \makeline 0 0 \makeline 0 1 \makeline 1 0 \makeline 0 2 \makeline 2 0 }$}.
Let $g' = f[g,\overline g]$, now $(g',\overline{g'})(b,c) \in \smallallpic b a b a$.
\bi
\item
If $(f,\overline{f})(c,b) \in \smallcrosspic c b c b$, then
\bi
\item
if $(f,\overline{f})(a,c) \in \smallflipmispic b a b a$, then $\psi = \{f \mapsto \frac 1 2, \overline{f} \mapsto \frac 1 2\} \in \fpol(\Gamma,\nu)$ and $a \domps\psi c$,
\item
otherwise $(f,\overline{f})(a,c) = (b,b)$.
Let $f' = f[f,\overline{f}]$ and note that $f'|_{\{a,b\}} = \pr_1$, $f'|_{\{b,c\}} = \pr_1$ and $f'(a,c)=f'(c,a)=b$.
With $g'' = g'[g',\overline{g'}]$ and $f'' = g''[\pr_1,f']$ it holds that $\psi = \{f'' \mapsto \frac 1 2, \overline{f''} \mapsto \frac 1 2\} \in \fpol(\Gamma,\nu)$ and $a \domps\psi c$.
\ei
\item
If $(f,\overline{f})(c,b) \in \smallallpic b a b a$, then
\bi
\item
if $(f,\overline{f})(a,c) \in \smallflipmispic b a b a$, then $\psi = \{f \mapsto \frac 1 2, \overline{f} \mapsto \frac 1 2\} \in \fpol(\Gamma,\nu)$ and $a \domps\psi c$,
\item
otherwise $(f,\overline{f})(a,c) = (b,b)$.
Let $f' = f[f,\overline{f}]$ and note that $f'|_{\{a,b\}} = \pr_1$, $f'(b,c),f'(c,b) \in \{a,b\}$ and $f'(a,c)=f'(c,a)=b$.
With $f'' = f'[\pr_1,f']$ it holds that $\psi = \{f'' \mapsto \frac 1 2, \overline{f''} \mapsto \frac 1 2\} \in \fpol(\Gamma,\nu)$ and $a \domps\psi c$.
\ei
\item
If $(f,\overline{f})(c,b) \in \smallcrosspic c a c a$, then $f' = f[f,\overline{f}]$ satisfies $f'|_{\{a,b\}} = \pr_1$,
$(f,\overline{f})(a,c) \in \smallallpic b a b a$ and $(f,\overline{f})(b,c) \in \smallallpic b a b a$, so the previous case applies.
\item
Otherwise $f \in \upa c b$.
\bi
\item
If $(f,\overline{f})(a,c) \in \smallmispic b a b a$, then assume \wlg $f(c,a) = b$.
Let $f'=g'[\pr_1,f]$.
Now $f'|_{\{a,b\}}$ is a projection, $(f',\overline{f'})(a,c) \in \smallallpic b a b a$ and $f' \not\in \upa c b$, so we may use one of the previous cases.
\item 
Otherwise $f \in \doa c a$.
Since $\{a,c\} \not\in \Gamma$, there is $i \in \pol(\Gamma)$ \st $i(c,a)=b$.
Note that $f' = i(\pr_1,f)$ satisfies $(f',\overline{f'})(a,c)=(a,b)$.
This takes us to the previous case.
\ei
\ei

\subsection{$\{a,b\} \not\in \Gamma$, $\{a,c\} \in \Gamma$, $\{b,c\} \not\in \Gamma$}
By \lemref{lem:fpol} there is $\omega \in \fpol(\Gamma,\nu)$ and $g \in \supp(\omega)$ \st $g \in \doa b a$.
By \lemref{lem:hascross}, $\smallcrosspic c a c a \in \Gamma$, or $(\Gamma,\nu)$ is of type \typegmc\@.
Unless $a \domps\omega b$ every $i \in \supp(\omega)$ is conservative on $\{a,b\}$.
Since $\{b,c\} \not\in \Gamma$ we also have $h \in \pol(\Gamma)$ \st $h(b,c)=a$.

Let $R$ be the relation generated by $\pol(\Gamma)$ from $\smallcrosspic c b c b$.
Since $h(b,c)=a$, $R$ can not equal $\smallmispic c b c b$ or $\smallcrosspic c b c b$.
Neither does $R$ equal $\mypic{\makeline 2 1 \makeline 1 2 \makeline 0 2 \makeline 2 0}$
or $\mypic{\makeline 2 2 \makeline 2 1 \makeline 1 2 \makeline 0 2 \makeline 2 0}$.
To see this note that;
in the fist case, since $\{c\} \in \Gamma$, it follows that $\{a,b\} \in \close{\Gamma}$, a contradiction,
and, in the second case there can not be a ternary operation $m \in \pol(\Gamma)$ that is arithmetic on $\{a,c\}$, this contradicts that $\{a,c\}$ is a cross-pair.
Hence, there must be $f \in \pol(\Gamma)$ \st $(f,\overline f)(b,c) \in \smallallpic b a b a$.

If $g \not\in \upa c b$,
then $\psi = \{ g \mapsto \frac 1 2, \overline{g} \mapsto \frac 1 2 \} \in \fpol(\Gamma,\nu)$ and $a \domps\psi b$.
We therefore assume $\pol(\Gamma) \cap \doa b a \subseteq \upa c b$.
This means that $\psi = \{ g \mapsto \frac 1 4, \overline{g} \mapsto \frac 1 4, f \mapsto \frac 1 4, \overline{f} \mapsto \frac 1 4 \} \in \fpol(\Gamma,\nu)$. So either $f,g$ are complementary and $(\Gamma,\nu)$ is of type \typegwtp, or $a \domps\psi b$ or $c \domps\psi b$.

\subsection{$\{a,b\} \not\in \Gamma$, $\{a,c\} \not\in \Gamma$, $\{b,c\} \in \Gamma$}
By \lemref{lem:fpol} there is $\omega \in \fpol(\Gamma,\nu)$ and $f',g \in \supp(\omega)$ \st $(f',\overline{f'})(a,c) \in \smallallpic b a b a$ and $g \in \doa b a$. 
By \lemref{lem:hascross}, $\smallcrosspic c b c b \in \Gamma$, or $(\Gamma,\nu)$ is of type \typegmc\@.
Assume \wlg that $g|_{\{b,c\}}=\pr_2$.
Let $f$ be any operation in $\pol(\Gamma) \cap \doa c a$.
Such an operation must exist.
Note that $h = g[f',\overline{f'}]$ satisfies $h(a,c)=h(c,a) \in \{a,b\}$. If $h \not\in \doa c a$, then $g[ g[ \pr_1, h ], \overline{g[ \pr_1, h ]} ] \in \doa c a$.
\bi
\item
If $(f,\overline{f})(a,b) \in \smallallpic b a b a$, then 
$\psi = \{ f \mapsto \frac 1 4, \overline{f} \mapsto \frac 1 4, g \mapsto \frac 1 4, \overline{g} \mapsto \frac 1 4\} \in \fpol(\Gamma \cup \{ \{a,b\} \},\nu)$.
Unless $a \domps\psi c$ we have $g \in \upa c a$, this means that $f,g \in \pol(\Gamma \cup \{ \{a,b\}, \{a,c\} \})$, so either
$f,g$ are complementary and $(\Gamma \cup \{ \{a,b\}, \{a,c\} \},\nu)$ is of type \typegwtp, or $a \domps\psi b$.

\item
If $(f,\overline{f})(a,b) \in \smallcrosspic c a c a$, then $f' = f[f,\overline{f}] \in \doa c a \doa b a$, so the previous case applies.

\item
If $(f,\overline{f})(a,b) \in \smallcrosspic c b c b$, then assume \wlg that $f(a,b)=c$.
It is easily checked that $g[ f[\pr_1,f], \overline{f[\pr_1,f]} ] \in \doa b a \doa c a$, so the first case applies.

\item
If $(f,\overline{f})(a,b) = (c,c)$, then assume \wlg $f|_{\{b,c\}} = \pr_1$. 
Once more we have $g[ f[ \pr_1 ,f ], \overline{f[ \pr_1 ,f ]} ] \in \doa b a \doa c a$, so again the first case applies.
\ei

\subsection{$\{a,b\} \not\in \Gamma$, $\{a,c\} \not\in \Gamma$, $\{b,c\} \not\in \Gamma$}
By \lemref{lem:hascross}, $R = \mypic{ \makeline 2 0 \makeline 1 1 \makeline 0 2 } \in \Gamma$, or $(\Gamma,\nu)$ is of type \typegmc\@.
In what follows we assume that $(\Gamma,\nu)$ is not of type \typegmc\@.
A consequence of this is that, for every binary function $f \in \pol(\Gamma)$, it holds that $f(a,c)=f(c,a)=b$ or $f|_{\{a,c\}}=\pr_1$ or $f|_{\{a,c\}}=\pr_2$.

By \lemref{lem:fpol} there is $\omega \in \fpol(\Gamma,\nu)$ and $g, h, i \in \supp(\omega)$ \st $(g,\overline{g})(a,c) \in \smallallpic b a b a$, $h \in \doa b a$ and $(i,\overline{i})(b,c) \in \smallallpic b a b a \cup \smallcrosspic c a c a$. 
Since $R \in \Gamma$ this means that $g(a,c)=g(c,a)=b$, $2 \nu(b) \le \nu(a)+\nu(c)$ and $h \in \upa c b$.
Hence, $h' = g[h,\overline{h}] \in \upa a b \upa c b$ and $h'(a,c)=h'(c,a)=b$.
\bi
\item
If $(i,\overline{i})(b,c) \in \smallcrosspic c a c a$, then $i' = h'[i,\overline{i}] \in \doa c b$ and $i'(a,c)=i'(c,a)=b$.
\item
If $(i,\overline{i})(b,c) \in \smallflipmispic b a b a$, then $i'' = h'[i,\overline{i}]$ satisfies $i''(b,c)=i''(c,b)=a$ and $i''(a,c)=i''(c,a)=b$.
Since $\{a,c\}\not\in\Gamma$ there is some binary $p \in \pol(\Gamma)$ \st $p(a,c)=b$.
This means that $i'=p[i'',h'] \in \doa c b$ and $i'(a,c)=i'(c,a)=b$.
\item
Otherwise $i \in \doa c b$.
Here $i'=g[i,\overline{i}]$ satisfies $i' \in \doa c b$ and $i'(a,c)=i'(c,a)=b$.
\ei
Since $i'$ must preserve $R$ it also holds that $i' \in \upa b a$.
By $i',h'$ we see that $(\Gamma,\nu)$ is of type \typebsm\@.

\section {Proof of \lemref{lem:fpol}}
\label{s:lem:fpol}
We need the following variant of Motzkin's Transposition Theorem that is easy to derive from \thmref{thm:motzkin}.
\begin{theorem}
\label{thm:motzkin:2}
For any $A \in \bbbq^{m \times n}$ and $B \in \bbbq^{p \times n}$, exactly one of the following holds:
\bi
\item $Ax \le 0$, $Bx < 0$ for some $x \in \qplus^n$
\item $A^Ty + B^Tz \ge 0$ and $z \ne 0$ for some $y \in \qplus^m$ and $z \in \qplus^p$
\ei
\end{theorem}

Let $\sigma : D \to \qplus \cup \{\infty\}$ be a function in $\eclose{\Gamma,\Delta}$ and for $\nu \in \Delta$ let $D_{\nu} = \{x \in D : \nu(x)<\infty\}$.
Let $\Omega_1 = \{f \in \binopers \cap \pol(\Gamma) : \nu(f(x,y))<\infty \text{ for every } \nu \in \Delta \text{ and } x,y \in D_{\nu}\}$ and $\Omega_2 = \{f \in \Omega_1 : \{f(a,b),f(b,a)\} \ne \{a,b\} \text{ and } \sigma(f(a,b))+\sigma(f(b,a)) \le \sigma(a)+\sigma(b) \}$.
Assume there are $p_{\nu,(x,y)} \in \qplus$ for $\nu \in \Delta$ and $(x,y) \in D_{\nu}^2$ \st 
\bna
\label{ineq:1}
&\ssum{\nu \in \Delta, (x,y) \in D_{\nu}^2} p_{\nu,(x,y)} \nu(g(x,y))
\ge \ssum{\nu \in \Delta, (x,y) \in D_{\nu}^2} p_{\nu,(x,y)}  \nu(\pr_i(x,y)), \quad g \in \Omega_1, i\in[2],\\
\label{ineq:2}
&\ssum{\nu \in \Delta, (x,y) \in D_{\nu}^2} p_{\nu,(x,y)} \nu(g(x,y))
>   \ssum{\nu \in \Delta, (x,y) \in D_{\nu}^2} p_{\nu,(x,y)}  \nu(\pr_i(x,y)), \quad g \in \Omega_2, i\in[2].
\ena
We will show that in this case we have $\smallcrosspic{a}{b}{a}{b} \in \wclose{\Gamma,\Delta}$.
Create an instance $I$ of \minhom$(\Gamma,\Delta)$ with variables $D^2$, and objective 
\ba
m(\phi) &= \ssum{\nu \in \Delta, (x,y) \in D_{\nu}^2} p_{\nu,(x,y)} \nu(\phi(x,y)) + \epsilon \ssum{\nu \in \Delta, (x,y) \in D_{\nu}^2} \nu(\phi(x,y)),
\ea
where $\epsilon > 0$ is choosen small enough so that $\phi \in \argmin_{\phi' \in \Omega_1} m(\phi')$ implies
$\phi \in \argmin_{\phi' \in \Omega_1} \sum_{\nu \in \Delta, (x,y) \in D_{\nu}^2} p_{\nu,(x,y)} \nu(\phi(x,y))$.
Such a number $\epsilon$ can always be found.
Note that a solution $\phi$ to $I$ with finite measure is a function $D^2 \to D$ \st $\nu(\phi(x,y))<\infty$ for every $\nu \in \Delta$ and $(x,y) \in D_{\nu}^2$.

Pick, for every $g \in \binopers \setminus \pol(\Gamma)$, a relation $R_g \in \Gamma$ \st $g$ does not preserve $R_g$.
Add for each pair of tuples $t^1,t^2 \in R_g$ the constraint $(( (t^1_1,t^2_1), \dots, (t^1_{\ar R_g},t^2_{\ar R_g}) ), R_g)$.
Now a solution to $I$ is a binary polymorphism of $\Gamma$.
Hence, if $\phi$ is a solution to $I$ with finite measure, then $\phi \in \Omega_1$.
Clearly $\pr_1$ and $\pr_2$ satisfies all constraints and are solutions to $I$ with finite measure.
By \eqref{ineq:1} the projections are also optimal solutions.
By \eqref{ineq:2} any function in $\Omega_2$ has a larger measure than $\pr_1$ or $\pr_2$.
Hence, if $P = \{(\phi(a,b),\phi(b,a)): \phi \in \optsol(I)\}$ we have $\argmin_{(x,y)\in P} (\sigma(x)+\sigma(y)) = \smallcrosspic a b a b$, and therefore $\smallcrosspic a b a b \in \wclose{\Gamma,\Delta}$.

This of course means that the system \eqref{ineq:1}+\eqref{ineq:2} can not be satisfied.
We can write \eqref{ineq:1}+\eqref{ineq:2} as
\begin{flalign*}
&& \ssum{\nu \in \Delta, (x,y) \in D_{\nu}^2} p_{\nu,(x,y)} (\nu(\pr_i(x,y))-\nu(g(x,y))) &\le 0, &\mathllap{g \in\Omega_1, i \in [2],}\\
&& \ssum{\nu \in \Delta, (x,y) \in D_{\nu}^2} p_{\nu,(x,y)} (\nu(\pr_i(x,y))-\nu(g(x,y))) &<   0, &\mathllap{g \in\Omega_2, i \in [2]}.
\end{flalign*}
If this system lacks a solution $p_{\nu,(x,y)} \in \qplus$ for $\nu \in \Delta, (x,y) \in D_{\nu}^2$, then by \thmref{thm:motzkin:2}, there are $z_{i,j,g} \in \qplus$ for $i,j \in [2]$, $g \in \Omega_1$ \st 
\begin{flalign*}
&& \ssum{i \in [2], j \in [2], g \in \Omega_j} (\nu(\pr_i(x,y))-\nu(g(x,y))) z_{i,j,g} &\ge 0, &\mathllap{\nu \in \Delta, (x,y) \in D_{\nu}^2,}
\end{flalign*}
where $z_{i,2,g} > 0$ for some $i \in [2], g \in \Omega_2$.
So
\begin{flalign*}
&& \ssum{i \in [2], j \in [2], g \in \Omega_j} (\nu(x)+\nu(y)-\nu(g(x,y))-\nu(g(y,x))) z_{i,j,g} &\ge 0,\quad\quad\quad\quad\quad &\mathllap{\nu \in \Delta, (x,y) \in D_{\nu}^2,}
\end{flalign*}
and with $z_{j,g} = z_{1,j,g} + z_{2,j,g}$ we have
\begin{flalign*}
&& \ssum{j \in [2], g \in \Omega_j} (\nu(x)+\nu(y)-\nu(g(x,y))-\nu(g(y,x))) z_{j,g} &\ge 0,\quad\quad\quad\quad\quad &\mathllap{\nu \in \Delta, (x,y) \in D_{\nu}^2,}
\end{flalign*}
Let $z'_{j,g} = z_{j,g} + z_{j,\overline g}$ and note that
\begin{flalign*}
&& \ssum{j \in [2], g \in \Omega_j} (\nu(x)+\nu(y)-\nu(g(x,y))-\nu(g(y,x))) z'_{j,g} &\ge 0,\quad\quad\quad\quad\quad &\mathllap{\nu \in \Delta, (x,y) \in D_{\nu}^2,}
\end{flalign*}
and therefore, since $z'_{j,g} = z'_{j,\overline g}$ for $j \in [2]$,
\begin{flalign*}
&& \ssum{j \in [2], g \in \Omega_j} (\nu(x)+\nu(y)-2\nu(g(x,y))) z'_{j,g} &\ge 0, &\mathllap{\nu \in \Delta, (x,y) \in D_{\nu}^2,}
\end{flalign*}
Note that by construction $z'_{2,g} > 0$ for some $g \in \Omega_2$.
We can rewrite this system into
\begin{flalign*}
&& \ssum{g \in \binopers} \omega(g) \nu(g(x,y)) &\le \frac{1}{2} ( \nu(x) + \nu(y) ), &\mathllap{\nu \in \Delta, (x,y) \in D_{\nu}^2,}
\end{flalign*}
by defining $\omega : \binopers \to \qplus$ as
\ba
\omega(g) = 
\frac{ \chi_{\Omega_1}(g) 2z'_{1,g} + \chi_{\Omega_2}(g) 2z'_{2,g} }
     {\sum_{g \in \binopers} (\chi_{\Omega_1}(g) 2z'_{1,g} + \chi_{\Omega_2}(g) 2z'_{2,g} )}.%
\footnotemark
\ea%
\footnotetext{Here $\chi_\Omega : \binopers \to \{0,1\}$ is the indicator function for the set $\Omega$.}
Clearly $\omega \in \fpol(\Gamma,\Delta)$ and $\omega(g) > 0$ for some $g \in \Omega_2$.

\section {Proof of \lemref{lem:hasconst}}
\label{s:hasconst}

Let $D = \{a,b,c\}$ and $\nu : D \to \qplus$ be \st $\nu(a) < \nu(b) < \nu(c)$.
We assume that $(\Gamma,\nu)$ is a min-core.
Let:
\ba
\gamma_1 = \mypic{ \makeline 2 0 \makeline 0 2 },\,
\gamma_2 = \mypic{ \makeline 1 0 \makeline 0 1 },\,
\gamma_3 = \mypic{ \makeline 1 2 \makeline 2 1 },\,
\gamma_4 = \mypic{ \makeline 2 0 \makeline 1 2 },\,
\gamma_5 = \mypic{ \makeline 1 0 \makeline 0 2 },\,
\gamma_6 = \mypic{ \makeline 2 0 \makeline 1 1 },\,
\gamma_7 = \mypic{ \makeline 2 0 \makeline 1 1 \makeline 0 2 }.
\ea

\begin{lemma}
\label{lem:hascross}
If $(\Gamma,\nu)$ is not of type \typegmc, then $\gamma_i \in \wclose{\Gamma,\nu}$, for some $i \in [7]$.
\end{lemma}
\begin{proof}
We will make use of the following fact~\cite[Lemma~5.6]{maxonesgen}:
If $(\Gamma,\nu)$ is of type \typegmc, then for every $R \in \Gamma$ we have $(\min_\nu \pr_1(R), \dots, \min_\nu \pr_{\ar(R)}(R) ) \in R$.
We say that a relation $R$ is generalised min-closed if this property is satisfied.

If $\Gamma$ is not of type \typegmc, then there is $R \in \Gamma$ that is not generalised min-closed.
Consider first the case when $\ar(R)=2$.
Let $w_1$ be the $\nu$-minimal element in $\pr_1(R)$, and $w_2$ be the $\nu$-minimal element in $\pr_2(R)$.
Let $q_1$ be the $\nu$-minimal element in $\{ x : (x,w_2) \in R\}$, and $q_2$ be the $\nu$-minimal element in $\{ y : (w_1,y) \in R\}$.
Set $\alpha = \frac{\nu(q_2)-\nu(w_2)}{\nu(q_1)-\nu(w_1)}$.
Now either $\argmin_{(x,y) \in R} \alpha \nu(x) + \nu(y)$ or its inverse is one of $\gamma_1,\dots,\gamma_7$.
This establishes the claim for $\ar(R)=2$.
Assume it holds also for every relation $R$ with $\ar(R)<m$.
Let $R_i = \{ (x_1,\dots,x_m) \in R : x_i \text{ is $\nu$-minimal in } \pr_i(R) \}$.
If $R_i$ is not generalised min-closed for some $i \in [m]$, then $\pr_{[m]-i}(R)$ is not generalised min-closed, so the result follows from the inductive hypothesis.
Otherwise $(q,w_2,\dots,w_m), (w_1,r,w_3,\dots,w_m) \in R$ where $w_i$ is the $\nu$-minimal element in $\pr_i(R)$ and $q$, $r$ are not $\nu$-minimal elements in $\pr_1(R)$ respective $\pr_2(R)$.
This means that $P = \{ (x,y) : (x,y,z_3,\dots,z_m) \in R \text{ and $z_i$ is $\nu$-minimal in } \pr_i(R)\}$ is not generalised min-closed, so again, the result follows from the inductive hypothesis.
\hfill\qed
\end{proof}

By \lemref{lem:hascross} we know that $\gamma_k \in \wclose{\Gamma,\nu}$, for some $k \in [7]$.
This immediately yields two constants: $\pr_1(\argmin_{(x,y) \in \gamma_k} \nu(x))$ and $\pr_1(\argmin_{(x,y) \in \gamma_k} \nu(y))$.

We may \wlg assume that one of these constants is $a$ since we always have $\argmin_{x \in D} \nu(x) = \{a\} \in \wclose{\Gamma,\nu}$.
Assume that the second constant is $c$ (the arguments for the other case is analogous).
Let $f = \{ a \mapsto a, b \mapsto a, c \mapsto c \}$.
Since $\Gamma$ is a min-core, we have $f \not\in \pol(\Gamma)$.
We must therefore have a wittiness $P \in \Gamma$ \st $t \in P$ but $f(t) \not\in P$.
Let $P' = \{ x : (z_1,\dots,z_{\ar(P)}) \in P \}$ where $z_i = c$ if $t_i=c$, $z_i = a$ if $t_i=a$ and $z_i = x$ otherwise.
Clearly $\{b\} = \argmin_{x \in P'} \nu(x)$.

\end{document}